\documentclass[11pt]{article}

\usepackage{verbatim}
\usepackage{amsmath}
\usepackage{amssymb}
\usepackage{amsthm}
\usepackage{yhmath}
\usepackage{graphicx}
\usepackage{subfigure}
\usepackage{epstopdf}
\usepackage[usenames,dvipsnames]{color}
\usepackage[margin=1in]{geometry}

    \newtheorem{theorem}{Theorem}[section]
    \newtheorem{lemma}[theorem]{Lemma}
    \newtheorem{definition}[theorem]{Definition}
    
    \newtheorem{corollary}[theorem]{Corollary}

\newcommand{\diag}{\mathop{\mathrm{diag}}}

\title{Duals of Orphan-Free Anisotropic Voronoi Diagrams are Triangulations}

\author{Guillermo D. Canas$^{\dagger}$,   Steven J. Gortler$^{\ddagger}$ \\
\small  $\dagger$ - CBCL, McGovern Institute,  
Massachusetts Institute of Technology\\
\small  $\ddagger$ - School of Engineering and Applied Sciences, 
Harvard University\\
\small Email:  \tt guilledc@mit.edu, sjg@seas.harvard.edu}

\date{}

\begin{document}

\maketitle

\begin{abstract}
Given an anisotropic Voronoi diagram, we address the fundamental question of when its dual is embedded. 
We show that, by requiring only that the primal be orphan-free (have connected Voronoi regions), 
its dual is always guaranteed to be an embedded triangulation. 
Further, the primal  diagram and its dual have properties that parallel those of ordinary Voronoi diagrams:
the primal's vertices, edges, and faces are connected, and the dual triangulation has a simple, closed boundary. 
Additionally, if the underlying metric has bounded anisotropy (ratio of eigenvalues), 
the dual is guaranteed to triangulate the convex hull of the sites.   
These results apply to the duals of anisotropic Voronoi diagrams of \emph{any}
set of sites, so long as their Voronoi diagram is orphan-free. 
By combining this general result with existing conditions for obtaining orphan-free anisotropic Voronoi diagrams, 
a simple and natural condition for a set of sites to form an embedded anisotropic Delaunay triangulation follows. 
\end{abstract}

\newpage

\section{Introduction}

Voronoi diagrams and their dual Delaunay triangulations are fundamental constructions with
numerous associated guarantees, and extensive application in
practice~\cite{Delaunay1934,Edelsbrunner}. 
At their heart is the use of a distance between points, which in the original
version is taken to be Euclidean. 
This suggests that, by considering 
generalizations of the Euclidean
distance, it may be possible to obtain variants which can be well-suited to
a wider range of applications.  

Attempts in this direction have been met with some success. 
Power diagrams~\cite{power} generalize Euclidean distance by associating 
a {bias-term} 
to each 
site. The duals of these diagrams
are 
guaranteed to be embedded triangulations, in any number of
dimensions. 
Although this is a strict generalization of Euclidean distance, it is a somewhat 
limited one. The effect of the bias term is to locally enlarge or shrink the
region associated to each site, loosely-speaking ``equally in every
direction". It allows some freedom in choosing local scale, with no
preference for specific directions. 

Another way to generalize Voronoi diagrams is to endow Euclidean space with a
continuously varying Riemannian metric, and use its associated (geodesic) distance 
in the definition of the Voronoi diagram. 
This generalization carries a significant
amount of freedom, allowing for local scale and ``directionality" to be freely
specified at each point.

Despite being of potentially great interest, this latter approach has faced several
obstacles. 
The dual of such an 
anisotropic Voronoi diagram is, in general, 
\emph{not} an embedded triangulation, even if the primal 
is orphan-free 
(the Voronoi regions are connected), 
and may produce element inversions and edge crossings. 
An additional obstacle is that the geodesic distance 
is extremely expensive to compute in practice, rendering traditional methods
for constructing Voronoi diagrams impractical. 

Two notable attempts to resolve these problems have been proposed. 
Independently, Labelle and Shewchuk~\cite{LS}, and Du and Wang~\cite{DW} propose
two efficient approximations of geodesic distance from a site to any point in the
domain. 
Although their associated Voronoi diagrams are, in general, no longer orphan-free, 
Labelle and Shewchuk show that a set of sites exists 
with an orphan-free diagram, whose dual is embedded, in two dimensions. 
They accomplish this by proposing an iterative site-insertion algorithm
that, for any given metric, constructs one such set of sites. 
Some recent work making use of the above definitions are~\cite{Boissonnat:2008:LUA:1377676.1377724,BOISSONNAT-2007-488446}.

In this paper, we show that if a set of sites produces an orphan-free 
anisotropic Voronoi diagram, using the definition of~\cite{DW}, 
then its dual is always an embedded triangulation (or a polygonal mesh with convex faces in general), 
in two dimensions (Thm.~\ref{th:final}). 
This effectively states that, regardless of the sites' positions, if the primal 
is well-behaved, then the
dual is as well. 
Further, in a way that parallels the ordinary Delaunay case, the dual has no
degenerate elements (Lem.~\ref{lem:degen}), its elements (vertices, edges, faces) are unique (Cor.~\ref{uniqueVD}), and, under mild assumptions on the 
metric, is guaranteed to triangulate the convex hull of the sites (Thm.~\ref{thm:boundary}). 
Note that, while~\cite{LS} prove a property of the output of an 
algorithm, the results in this paper are fundamental properties of the diagrams themselves 
(independent of how the orphan-freedom of the primal was obtained).

In practice, we may combine our results with those of~\cite{avd}  
to conclude that duals of anisotropic Voronoi diagrams of appropriate $\epsilon$-nets 
are always embedded triangulations (Cor.~\ref{cor:enet}). 
This may be particularly useful in (asymptotically-optimal) function approximation applications, 
where we are often interested in 
constructing the anisotropic Delaunay triangulations of appropriate $\epsilon$-nets~\cite{enets,GruberOQ}, and was the initial motivation for the current work. 
Finally, we note that, as discussed in Sec.~\ref{sec:implementation}, in practice, 
algorithms for constructing anisotropic diagrams of the type discussed here may have a more numerical flavor  
(i.e.\  front-propagation, fast marching methods), in contrast with the more combinatorial nature of other Voronoi diagrams. 

\section{Setup}\label{sec:setup}

Given a finite set $V\subset\mathbb{R}^2$ of 
sites on the plane, 
a Voronoi 
diagram 
decomposes Euclidean space into regions, each distinguished by the site
its points are closest to. 
The notion of closeness is defined as follows:
consider a continuous metric (in coordinates: $Q:\mathbb{R}^2\rightarrow\mathbb{R}^{2\times 2}$, 
symmetric, positive definite) 
over two-dimensional Euclidean space and,
following~\cite{DW}, 
define the (asymmetric) ``distance" between a site
($v\in V$) and point ($p\in\mathbb{R}^2$) as
\[ D(v,p) = \left[(p-v)^t Q_p (p-v)\right]^{1/2} \]

The Voronoi region of site $v$ is the set of points no further from $v$ than
from any other site:
\[ R(v) = \{p\in\mathbb{R}^2 : D(v,p) \le D(w,p), \forall w\in V\} \]
Note that, because $V$ is a set, and not a multiset, there are no coincident sites.

This definition of anisotropic Voronoi diagram 
results in diagrams that take the usual shape, but
whose regions may have curved boundaries, and tend to be elongated along
certain directions, depending on the metric $Q$. 

Following the notation of~\cite{LS}, we say that a diagram is
\emph{orphan-free} if its regions $R(v)$ are connected, for all $v\in V$
(equivalently: each \emph{cell}, or connected component of a Voronoi region, 
{contains} its generating site).

An important distinction in this construction is that  the ``interfaces" 
in-between Voronoi regions do not, in general, have null measure 
(and points equidistant to three or more sites are not always isolated), 
unlike 
geodesic-distance diagrams~\cite{LL2000} 
(see~\cite{enets} Lemma 5.2), and those of~\cite{LS}. 

Note that, given one such degenerate diagram, any interior point $p$ of a Voronoi edge
$R(v)\cap R(w)$
of non-zero measure 
has an 
open neighborhood $N_p\subset R(v)\cap R(w)$. 
By adding to $Q$ an arbitrarily
small perturbation supported on $N_p$, an orphan is created.
That is, with respect to the metric, any degenerate diagram is
arbitrarily close to a diagram that has orphans, and therefore the
non-degeneracy requirement is only slightly more restrictive than the
orphan-freedom one.

We begin with three simple lemmas, proved in Appendix A, 
which provide an introduction to some of the techniques used in the sequel. 
The first is a simple extension of Lemma 2.1 of~\cite{DW}, and the second follows directly from the continuity of $Q$ and the fact that every site is {strictly} closer to itself than to all other sites. 

\begin{lemma}\label{lem:midpoint}
	Given two sites $v,w\in V$, $v\ne w$, the only point equidistant to $v,w$ in their supporting
line is the midpoint $(v+w)/2$. 
\end{lemma}

\begin{lemma}\label{lem:interior}
    Every site is an interior point of its corresponding Voronoi region. 
\end{lemma}

\begin{lemma}\label{lem:sc}
Every Voronoi region 
of an orphan-free anisotropic Voronoi 
diagram in 
$\mathbb{R}^2$ 
is simply connected. 
\end{lemma}

\noindent{\bf Primal diagram}. 
The Voronoi diagram is a collection of Voronoi regions $R(v)$, one for every site in $V$, as well as a structure 
induced by the sets of points closest and equidistant to two or more sites. 
Let $\tilde{P}$ to be the primal Voronoi diagram embedded on the sphere
$\mathbb{S}^2$ by stereographically projecting the plane onto the 
punctured sphere and
completing it with a ``point at infinity" $p_\infty$.
$\tilde{P}$ includes not just regions, but also embedded edges $\tilde{E}_{v w}$,
which are \emph{connected} sets of points closest and equidistant to two sites $v,w$, as well
as vertices (points closest and equidistant to three or more sites). 
Since we are assuming that sets of  equidistant points to two sites have
null measure, and that
points equidistant to three or more sites are isolated, 
we can consider the abstract mesh  $P=(V_p,E_p,F_p)$ with structure
derived from 
$\tilde{P}$, where $V_p,E_p,F_p$ are the vertices, edges, and faces of $P$.

A simple induction argument reveals that every edge in $E_p$ connects two
vertices in $V_p$. 
To see this, start with an orphan-free diagram with two sites. The only edge
of $\tilde{E}$ is the unbounded set of points equidistant to the sites, passing through
their midpoint (Lem.~\ref{lem:midpoint}), and with endpoints at $p_\infty$. 
The insertion of an additional site can only split existing edges of
$\tilde{E}$ at points
where two edges cross (a vertex by definition), and thus the
original edge is split into pieces whose endpoints are now either $p_\infty$ or 
an edge crossing (a vertex). 

Clearly, $\tilde{P}$ is a planar embedded diagram since two edges 
cannot meet except at vertices:  
any point $p$ where edge $D(u,p)=D(v,p)$ crosses 
edge $D(v,p)=D(w,p)$ is equidistant to three sites, $u,v,w$, and thus a vertex. 
Therefore $P$ is planar. \\


\noindent{\bf Dual of an orphan-free Voronoi diagram}.
The dual of the primal  mesh structure $P$ is another abstract 
mesh ${G}=(V,{E},F)$, with one vertex per site, 
and edges connecting adjacent regions of $\tilde{P}$. 
Its face structure is derived from that of $P$ by duality   
and, in particular, $(V,E)$ is planar, since $(V_p,E_p)$ is planar. 
Using the embedded primal $\tilde{P}$, we can define an embedding $\tilde{G}$ of $G$ 
on the plane, with vertices coinciding with sites (and curved edges). 
Although $\tilde{P}$ was defined on the sphere, we embed $\tilde{G}$ on
the plane instead, and therefore we exclude (from both $G$ and $\tilde{G}$) the unbounded face dual to the point at infinity $p_\infty$. 

Since $(V,E)$ is planar, 
$G$ derives the following properties from the primal:
\begin{itemize}
\item[(i)]
Every edge in $E$ is incident to two dual faces, 
except for \emph{boundary edges} (whose corresponding primal edge in
$\tilde{P}$ connects unbounded regions), which are only incident to one dual face. 
\item[(ii)]To every dual face in $F$, connecting vertices $v_1,\dots,v_m$, corresponds a
primal vertex $c\in R(v_1)\cap\dots\cap R(v_m)$ of $\tilde{P}$, 
which is equidistant to the corresponding sites $v_1,\dots,v_m$. 
\end{itemize}
The property (ii) requires that the Voronoi diagram be orphan-free, as noted in the remark after Thm.~\ref{th:ece}.

Notice that either $\tilde{G}$ or $\tilde{P}$ could be multigraphs (e.g.\ $\tilde{P}$ would be a multigraph if an edge $\tilde{E}_p$ of $\tilde{P}$ has multiple connected components). 
However, the following lemma ensures that both are simple graphs.

\begin{lemma}\label{lem:simple}
$\tilde{G}$ and $\tilde{P}$ are simple (have no multi-edges or self-loops). 
\end{lemma}
\begin{proof}
See Appendix D. 
\end{proof}

By its definition, an edge of $\tilde{P}$ is a connected set of points equidistant
and closest to two sites $v,w$. Since $\tilde{P}$ has no multi-edges, this implies the following:

\begin{corollary}\label{cor:Pij}
The set 
of points closest and equidistant to two sites of an orphan-free diagram is connected. 
\end{corollary}

Property (ii) above can be restated in a more useful way 
that parallels an equivalent property of duals of ordinary Voronoi diagrams. 

\begin{definition}[{\bf Empty circum-ellipse property}]\label{def:ece}
	A face of $G$ incident to vertices $v_1,\dots,v_m$ satisfies the empty circum-ellipse (ECE) property if there is
\emph{some} ellipse that circumscribes the sites corresponding to $v_1,\dots,v_m$,
and contains no site  in its interior. 
\end{definition}

Note that \emph{any} empty circum-ellipse serves as witness to the ECE property. 
We show that the ellipse centered at a Voronoi vertex $c$, with axis given by $Q_c$, is always an empty circum-ellipse of its dual polygon.

\begin{theorem}\label{th:ece}
	Every face of $G$ satisfies the empty circum-ellipse 
	property. 
\end{theorem}
\begin{proof}
The proof is constructive. 
To every face in $G$ connecting vertices with corresponding sites $v_1,\dots,v_m$, 
	by property (ii), 
	corresponds a 
	Voronoi vertex
$c\in R(v_1)\cap\dots\cap R(v_m)$ of $\tilde{P}$ that is equidistant to $v_1,\dots,v_m$. 
If $\mu=D(v_1,c)=\dots=D(v_m,c)$,
then the open ellipse 
\[\theta=\{p\in\mathbb{R}^2 : D(p,c) < \mu\} = \{p\in\mathbb{R}^2 :
(p-c)^t Q_{c} (p-c) < \mu^2 \} \]
circumscribes the $v_1,\dots,v_m$ and does not contain any site (otherwise, $c$ would be closer to that site than to $v_1,\dots,v_m$, 
and therefore not in $R(v_1)\cap\dots\cap R(v_m)$, a contradiction). 
\end{proof}

\begin{figure}[ht]
\centering
\subfigure[]{
\includegraphics[height=3.5cm]{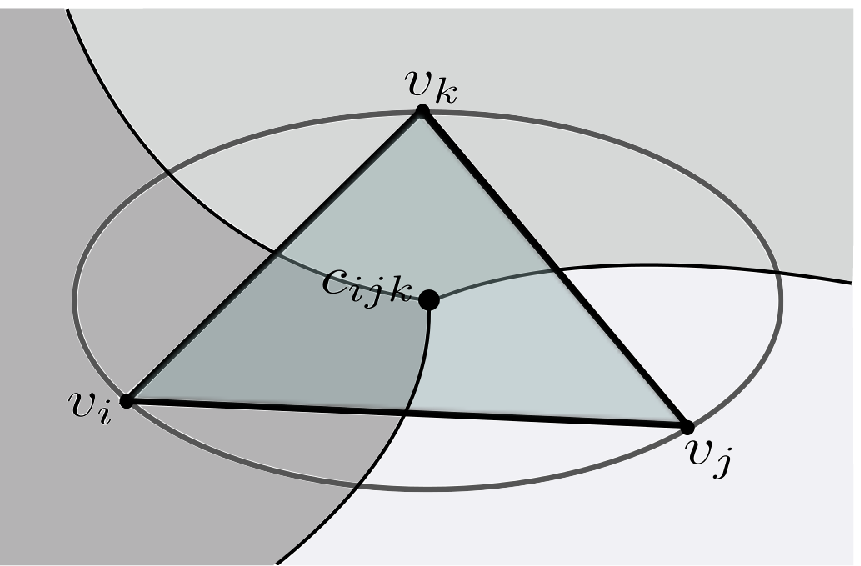}
\label{fig:f4a}
}
\subfigure[]{
\includegraphics[height=3.5cm]{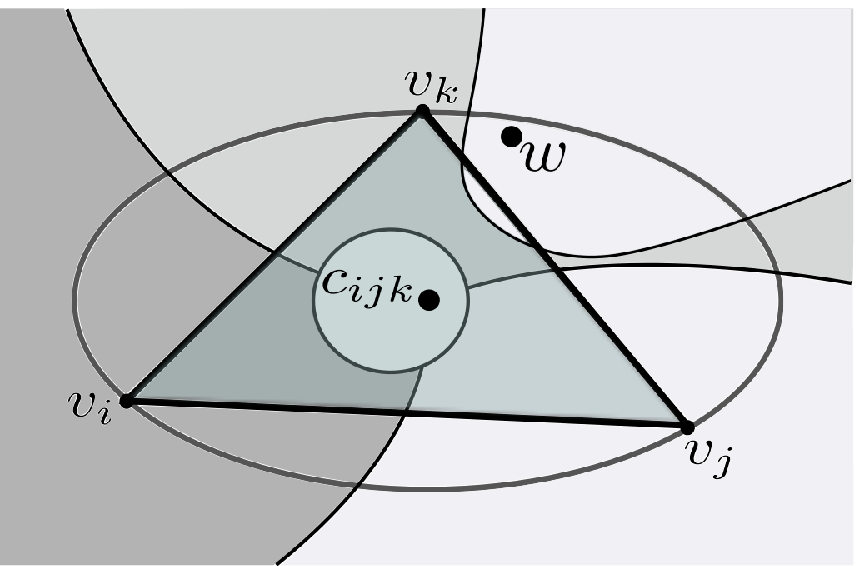}
\label{fig:f4b}
}
\label{fig:f4}
\caption{In an orphan-free diagram (a), every face  satisfies the
empty-circum ellipse (ECE) property. If the diagram is not orphan-free, then
some faces may not satisfy the ECE property.}
\end{figure}

Although this may not at first seem apparent,
the above property uses the fact that the primal diagram is orphan free. 
The critical fact in the proof of Thm.~\ref{th:ece} is not that
there is a Voronoi vertex $c$ equidistant to $v_1,\dots,v_m$, but rather that 
it is $c\in R(v_1)\cap\dots\cap R(v_m)$. 
Consider Figs.~\ref{fig:f4a} and~\ref{fig:f4b}, in which sites $v_i,v_j,v_k$ form a face of $G$. 
In~\ref{fig:f4a}, the diagram is orphan-free, and the fact that $c_{ijk}\in
R(v_i)\cap\ R(v_j)\cap R(v_k)$ implies Thm.~\ref{th:ece}. 
Figure~\ref{fig:f4b} shows part of a Voronoi diagram that has an orphan cell (corresponding
to site $w$) ``covering" $c_{ijk}$. 
If the dual is constructed such that the face $v_i,v_j,v_k$ is not in the dual  
then a ``hole" in the triangulation may occur. 
If the face $v_i,v_j,v_k$ is in the dual, 
then, 
despite the fact that $c_{ijk}$ is equidistant
to $v_i,v_j,v_k$, this $c_{ijk}$ is no longer in 
$R(v_i)\cap R(v_j)\cap R(v_k)$, but in the Voronoi region
of $w$. 
Since $c_{ijk}$ is closer to $w$ than to $v_i,v_j,v_k$, 
any open ellipse centered at $c_{ijk}$, and having
$v_i,v_j,v_k$ on its boundary, contains $w$.
Therefore, in Fig.~\ref{fig:f4b}, the ECE 
property does not hold for the dual
face of $c_{ijk}$.  
Ensuring that Thm.~\ref{th:ece} holds is one of the main reasons for requiring orphan-freedom. 


\section{Summary of Results and Outline}

The two main results in this paper assume that we are given an orphan-free anisotropic Voronoi diagram (in two dimensions) and 
prove that: 1) \emph{the dual is an embedded triangulation} (Thm.~\ref{th:final}), 
and 2) \emph{if the metric has bounded anisotropy (ratio of eigenvalues), the dual triangulates the convex hull of the sites} (Thm.~\ref{gamma}).

We prove a number of additional results. Possibly the most important shows that the elements of an orphan-free diagram (vertices, edges, faces), 
are connected sets, and therefore, in some sense, unique (Cor.~\ref{uniqueVD}). 
This is important, since this is a natural well-behave-ness condition on the primal, which, together 
with the above results, indicates that orphan-freedom is sufficient to guarantee well-behave-ness of both primal and dual. 

A smaller result of interest is: the dual faces satisfy an \emph{empty circum-ellipse} property (Thm.~\ref{th:ece}), which parallels the empty-circumcircle property of ordinary Voronoi diagrams, 
and could have further practical implications from the ones described here. 

The order of the proofs is, however, different from the one just stated. 
In effect, we begin by proving the most restrictive case: that the dual of an orphan-free diagram with metric of bounded anisotropy is embedded (Sec.~\ref{sec:interior}), 
and that its boundary is the boundary of the convex hull of the sites (Sec.~\ref{sec:boundary}). 
We then relax the bounded anisotropy condition (Sec.~\ref{sec:gen}), and show that, in the general case, we loose the convex hull property, but the dual remains embedded. 
Finally, once these results are established, we prove that vertices of the primal are unique (uniqueness of edges was proven in Sec.~\ref{sec:setup}, and uniqueness of primal faces is the same as orphan-freedom). 

We assume in the remainder of the paper that 
there are more than two sites, not all of which are colinear, and relegate the (considerably simpler) colinear case to Appendix B. 

%
%
%



\section{Dual of Orphan-free Diagram (Part I: boundary)}
\label{sec:boundary}

In this section, we assume that the metric has bounded anisotropy, and conclude that the boundary of the dual of an orphan-free diagram is the same as the boundary of the convex hull of the sites (and in particular is simple and closed). 
If $R_p\Lambda_p R_p^t$ is the eigendecomposition of $Q_p$ at $p\in\mathbb{R}^2$, with $\Lambda_p = diag[\lambda_2(p),\lambda_1(p)]$, $\lambda_2(p)\ge\lambda_1(p)>0$, and $R_p$ orthonormal, then we assume that there is some bound $\gamma$ on the anisotropy of $Q$, such that: $1\le \lambda_2(p)/\lambda_1(p) < \gamma^2$ for $p\in\mathbb{R}^2$. 
Note that this condition may commonly hold in practice, if the metric is sampled on a compact domain (and possibly extended to the plane by reusing sampled values only). 

We begin by defining $\bar{G}$ to be the straight-edge drawing of 
$\tilde{G}$ with vertices at the sites. 
For the moment, we assume that every Voronoi vertex is equidistant to no more than three sites, and therefore that all faces in $\bar{G}$ are triangles 
(we extend the results to the general case in Appendix G). 
We associate to $\bar{G}$ 
 a mapping (from the vertices, edges, and faces
of $\bar{G}$ to $\mathbb{R}^2$) for which, because all faces of $\bar{G}$ are triangles, it is well defined whether a point in
$\mathbb{R}^2$ belongs to any given face of $\bar{G}$, a fact that will be
used in the proofs of Sec.~\ref{sec:interior}. 
This mapping will be shown in Sec.~\ref{sec:interior} to be an embedding. 
In the sequel, it is assumed that $\bar{G}$ encompasses both the mesh
structure, and the mapping. 




The boundary vertices of $G$ are those whose corresponding primal regions in $\tilde{P}$ are unbounded, 
while boundary edges of $G$ connect boundary vertices. 
Note that this is a \emph{topological} property of $G$, rather than a geometric one (boundary elements of $G$ may, in principle, not lie in the boundary of the convex hull of the sites). 
For convenience, we call $B\subseteq E$ the set of boundary edges of $G$. 


The convex hull $\mathcal{CH}(V)$ of $V$ is the minimal (w.r.t.\ set
containment) convex set that contains $V$. 
For convenience, we name $W=\{w_i\in V : i=1,\dots,m\}$ the sites that are part of the boundary of
the convex hull $\mathcal{CH}(V)$, and order them 
in clock-wise order around $\mathcal{CH}(V)$. 
The boundary $\mathcal{B}$ of the convex hull is a simple circular chain 
$\mathcal{B} = \{(w_i, w_{i\oplus 1}) : i=1,\dots,m\}$. 
We prove that 
it is $B=\mathcal{B}$ 
(loosely speaking: the topological boundary of $G$, and the geometric boundary of its straight-edge embedding $\bar{G}$, are the same), 
which implies that $\bar{G}$ covers
the convex hull of the sites, and its boundary edges form 
a simple, closed polygonal chain. 
All the proofs of this section are in Appendix E. 

\begin{lemma}\label{boundary_easy}
 To every boundary edge $(v,w)$ of $G$ corresponds a segment in the boundary of $\mathcal{CH}(V)$. \emph{$[B\subseteq\mathcal{B}]$}
\end{lemma}

We now turn to the converse claim: that to every segment 
$(w_i,w_j)\in\mathcal{B}$ corresponds a boundary edge
$(w_i,w_j)\in B$ in $G$. 
Since $B$ is the set of boundary edges of $G$, whose primal edges in
$\tilde{P}$ are
unbounded, the claim is equivalent to proving that, to every segment
$(w_i,w_j)$ in the boundary of the convex hull corresponds an edge of the
dual $G$, whose primal edge $\tilde{E}_{w_i,w_j}$ is unbounded. 

The proof proceeds as follows. First, assume w.l.o.g.\  that the origin is in the interior of $\mathcal{CH}(V)$. 
Let  $C(\sigma)=\{x\in\mathbb{R}^2 : \|x\|=\sigma\}$ be a sufficiently large origin-centered circle. 
We define two  functions: 
\begin{eqnarray*}
	\pi : C(\sigma)\rightarrow\partial\mathcal{CH}(V),   \text{  }\text{  }\text{  }   \pi(p) = \underset{r\in\mathcal{CH}(V)}{\operatorname{\mathbf{argmin\ }}} D(r,p) \\
	\nu : \partial\mathcal{CH}(V)\rightarrow C(\sigma), \text{  }\text{  }\text{  }\text{  }\text{  }\text{  }    \nu(r) =  \sigma \cdot r / \|r\|
\end{eqnarray*}
$\pi$ simply projects every point in $C$ to its closest in $\mathcal{CH}(V)$ (with respect to the distance $D$), 
and $\nu$ projects a point back to $C$.

\begin{lemma}\label{lem:piC0}
$\pi(p) = \underset{x\in\mathcal{CH}(V)}{\operatorname{argmin}}D(x,p)$ is a continuous function in $\mathbb{R}^2$. 
\end{lemma}
\begin{proof}
We first prove by contradiction that $\pi(p)$ is unique. 
Let $M_p$ be the unique~\cite{avd}, symmetric positive definite square-root matrix such that $M_p^t M^t_p = Q_p$. 
Consider distinct
$q_1,q_2\in\mathcal{CH}(V)$ closest to $p$. 
Since $q_1\neq q_2$, by the convexity of the Euclidean norm, and the positive-definiteness of $M_p$, it is
\[ D((q_1+q_2)/2,p) = \|M_p\left( (q_1 + q_2)/2 - p)\right)\| < \| M_p\left(q_1- p\right)\| = \| M_p\left(q_2 - p\right)\| \]
Thus $p$ is closer to $(q_1+q_2)/2$ than to
$q_1,q_2$, a contradiction. Therefore the closest point $\pi(p)$ to $p$ is
unique. 

If $\pi$ were not continuous at $p$, then there is $\epsilon > 0$ such that 
for all $\delta > 0$ there is a point $\bar{p}=\bar{p}(\delta)$ such that $\|p-\bar{p}\| < \delta$
and $\|\pi(p) - \pi(\bar{p})\| \ge \epsilon$. 
Consider the sequence $\{\kappa_i = \pi(\bar{p}(1/i)) : i\in\mathbb{N}\}$ of points in
$\mathcal{CH}(V)$. Because $\mathcal{CH}(V)$ is compact, $\{\kappa_i\}$ 
has a
subsequence that converges to some $\kappa\in\mathcal{CH}(V)$. 
By continuity of 
$D$ (which follows from the continuity of $Q$), it is
\begin{eqnarray*}
	D(\pi(p),p) = \min_{x\in\mathcal{CH}(V)} D(x,p) = \displaystyle{  \lim_{i\rightarrow\infty} \min_{x\in\mathcal{CH}(V)} D(x,\bar{p}_i) } 
			= \displaystyle{ \lim_{i\rightarrow\infty} D(\kappa_i,\bar{p}_i) = D(\kappa, p)  }
\end{eqnarray*}
and therefore the closest point in $\mathcal{CH}(V)$ to $p$ is not unique, a
contradiction. 
\end{proof}

Clearly, $\nu$ is continuous in $\partial\mathcal{CH}(V)$. 
Note that, because $\mathcal{CH}(V)$ contains the origin, then, as shown in Fig.~\ref{fig:pinu}, $\nu$ projects 
every point $\pi(p)\in (w_i,w_j)$ on a segment of $\partial\mathcal{CH}(V)$, 
\emph{outwards} from the convex hull (and on the empty side of
$(w_i,w_j)$); that is, so that $\nu(\pi(p))\in H^{+}_{ij}\cap C$
($\nu(\pi(p))$ is in the empty half-space of $(w_i,w_j)$). 

\begin{figure}[htbp]
   \centering
   	\includegraphics[width=2.7in]{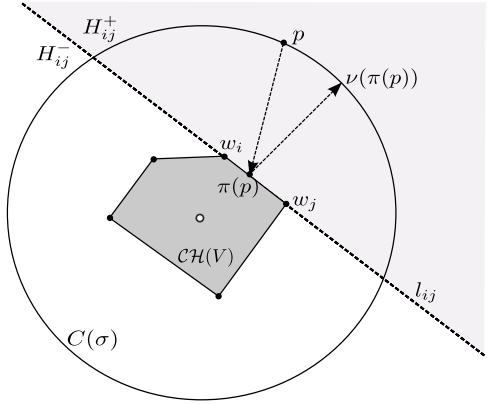} 
   \caption{The construction for the proof of Thm.~\ref{thm:boundary}.}
   \label{fig:pinu}
\end{figure}

Conveniently, Lem.~\ref{lem:mij} shows that, if $\pi(p)=(w_i+w_j)/2$, then
$p\in \tilde{E}_{w_i,w_j}$, and so the claim now reduces to showing that for each
segment $(w_i,w_j)$ of $\partial\mathcal{CH}(V)$, and for \emph{every} sufficiently large circle $C$, 
there is $p\in C$ with $\pi(p)=(w_i+w_j)/2$. Since this implies that $\tilde{E}_{w_i,w_j}$ is
unbounded, it means that the corresponding edge $(w_i,w_j)$ is in $B$. 

The proof is by contradiction. 
Lem.~\ref{lem:Sn} uses Brouwer's fixed point theorem to show that, for every segment $(w_i,w_j)$ of $\mathcal{B}$, 
if there were no $p\in C$ with $\pi(p)=(w_i+w_j)/2$, then 
the function $\nu\circ\pi:C\rightarrow C$ must have a point $q\in C$ such that
$\nu(\pi(q))=-q$, or equivalently such that 
$q$ is ``behind" the segment $(w_i,w_j)\in\partial\mathcal{CH}(V)$ to which it is
closest ($q\in H^{-}_{ij}$). 
On the other hand, Lem.~\ref{lem:contrad} shows that, for all sufficiently large circles $C$, no
point $q\in C$ can be closest to some segment
$(w_i,w_j)\in\partial\mathcal{CH}(V)$ that it is \emph{behind} of, creating a contradiction. 

%

%

As mentioned above, we are interested in identifying points in $C(\sigma)$
that are also in some primal edge $\tilde{E}_{w_i,w_j}$. The following lemma will be used
to show that, for sufficiently large $\sigma$, such a point $p$ can be
characterized by satisfying $\pi(p)=(w_i+w_j)/2$.

\begin{lemma}\label{lem:mij}
	If $\pi(p)=(w_i+w_j)/2$, with $p\in\mathbb{R}^2$,  $w_i,w_j\in W$, then
$D(w_i,p) = D(w_j,p)$. 
\end{lemma}

The following technical lemma is the key in constructing a contradiction by showing that, 
for sufficiently large circles $C(\sigma)$, no
point $q\in C(\sigma)$ can be closest to some segment
$(w_i,w_j)\in\partial\mathcal{CH}(V)$ that it is behind of ($q\notin H^{-}_{ij}$); 
where, as before, the open half space $H^{-}_{ij}$ is chosen to be 
the only of the two half spaces on either side of the supporting line $l_{ij}$ of $(w_i,w_j)$ 
such that $H^{-}_{ij}\cap V\neq\phi$.


\begin{lemma}\label{lem:contrad}
	There is $\rho$ such that, 
	for any segment $(w_i,w_j)\in\mathcal{B}$, with supporting line $l_{ij}$, 
	every $p\in H^{-}_{ij}$ with $\|p\| > \rho$
    whose closest point in $l_{ij}$ is $m_p\in\overline{w_i w_j}$ is 
	closer to a site in $V\setminus\{w_i,w_j\}$ than to $l_{ij}$. 
\end{lemma}


\begin{lemma}\label{lem:Sn}
	Every continuous function $F:\mathbb{S}^n\rightarrow\mathbb{S}^n$ that is not onto has a fixed point. 
\end{lemma}

\begin{lemma}\label{lem:hard}
 To every segment  in the boundary of
$\mathcal{CH}(V)$ corresponds a boundary edge of $G$. \emph{$[B\supseteq\mathcal{B}]$}
\end{lemma}

Therefore,

\begin{theorem}\label{thm:boundary}
If the metric $Q$ has bounded ratio of eigenvalues, then the boundary edges of $\bar{G}$ are the same as the boundary edges of the convex hull of $V$.  \emph{$[B=\mathcal{B}]$} 
\end{theorem}

\begin{corollary}\label{col:simple-boundary}
If the metric $Q$ has bounded ratio of eigenvalues, then the boundary edges of $\bar{G}$ form 
a simple, closed, convex polygonal chain. 
\end{corollary}

The ECE property is sufficient to show that no face in $\bar{G}$ is degenerate, a fact that will be used in the following section. 

\begin{lemma}\label{lem:degen}
	$\bar{G}$ has no degenerate (null area) elements. 
\end{lemma}

In Thm.~\ref{th:main}, we show that, even if Voronoi vertices are incident to more than three Voronoi regions (equidistant to more than three sites), 
every face in $\bar{G}$ is a (strictly) convex polygon, and therefore can be trivially triangulated (say in a fan arrangement). 
The resulting triangulation has no degenerate elements, and all its triangles satisfy the ECE condition (using the same witness ellipsoid as the convex polygon from which they are triangulated).

\section{Dual of Orphan-free Diagram (Part II: interior)}\label{sec:interior}

In this section, we assume that the (topological) boundary of $\bar{G}$ is simple and closed, and prove that $\bar{G}$ must be embedded.
The main argument in the proof 
uses  Theorems~\ref{th:ece} and~\ref{thm:boundary}, Cor.~\ref{col:simple-boundary}, as well as 
the theory of discrete one-forms on graphs, 
to show that there are no
``edge foldovers" in $\bar{G}$ 
(edges whose two incident faces are on the same side of its supporting line), 
and use this to conclude that $\bar{G}$ is embedded (Thm.~\ref{th:main}). 
As in 
Sec.~\ref{sec:boundary}, we assume that not all sites are colinear 
(the simpler colinear case was addressed in Appendix B). 
We distinguish between the sites $W\subseteq V$ that lie on the boundary of the convex hull, and the remaining, or
\emph{interior sites} ($V\setminus W$).

The following definition, from~\cite{1form}, 
assumes that, for each edge
$(v_i,v_j)$ in $G$, we 
distinguish the two opposing half-edges
$(v_i,v_j)$ and $(v_j,v_i)$.

\begin{definition}[Gortler et al.\ \cite{1form}]\label{def:1form}
A non-vanishing (discrete) one-form $\xi$  is an assignment of a real value
$\xi_{ij} \neq 0$ to each half edge $(v_i,v_j)$ in $G$, such that 
$\xi_{ji} = -\xi_{ij}$. 
\end{definition}


Since $\bar{G}$ has the same structure as $G$, we can construct a non-vanishing one-form
over $\bar{G}$ as follows. 
Given some unit direction vector $n\in\mathbb{S}^1$
(in coordinates $n=\left[n_1,n_2\right]^t$), 
we assign a real
value $z(v) = n^t v$ to each vertex $v$ in $\bar{G}$, and define 
$\xi_{ij} = z(v_i) - z(v_j)$, which clearly satisfies 
$\xi_{ji} = -\xi_{ij}$. The one-form, denoted by $\xi^n$, 
is non-vanishing if, for all edges $(v_i,v_j)\in E$, 
it is $\xi_{ij} = n^t (v_i - v_j) \neq 0$. 
That is, 
if the direction $n$ is not orthogonal to any edge. 
The set of edges has cardinality $|E| \le |V| (|V|-1)/2$, and in particular
it is finite. Therefore \emph{almost all} directions $n\in\mathbb{S}^1$ generate a non-vanishing one-form $\xi^n$.


Since $G=(V,E,F)$ is an planar graph with a well-defined face structure,
there is, for each face $f\in F$, a cyclically ordered set
$\partial f$ of half-edges round the face. 
Likewise, for each vertex $v\in V$, the set $\delta v$ of cyclically ordered
(oriented) half-edges emanating from each vertex is well-defined. 

\begin{definition}[Gortler et al.\ \cite{1form}]
Given 
non-vanishing one-form $\xi^n$ corresponding to $n\in\mathbb{S}^1$, 
the index of vertex $v$ with respect to $\xi^n$ is $\mathbf{ind}_{_{\xi^n}}(v) = 1 - \mathbf{sc}_{_{\xi^n}}(v) / 2$, where $\mathbf{sc}_{_{\xi^n}}(v)$ is the
number of sign changes of $\xi^n$ as one visits the half-edges of $\delta
v$ in order. \\
The index of face $f$ is $\mathbf{ind}_{_{\xi^n}}(f) = 1 - \mathbf{sc}_{_{\xi^n}}(f)/2$ where $\mathbf{sc}_{_{\xi^n}}(f)$ is the number
of sign changes of $\xi^n$ as one visits the half-edges of $\partial f$ in
order. 
\end{definition}

Note that, by definition, it is always $\mathbf{ind}_{_{\xi^n}}(v) \le 1$. 
A discrete analog of the Poincar\'e-Hopf index theorem relates 
the two indices above:

\begin{theorem}[Gortler et al.\ \cite{1form}]\label{lem:ph}
For any non-vanishing one-form $\xi^n$, it is 
\[ \displaystyle{\sum_{v\in V} \mathbf{ind}_{_{\xi^n}}(v) + \sum_{f\in F} \mathbf{ind}_{_{\xi^n}}(f) = 2} \]
\end{theorem}

Note that this follows from Theorem 3.5 of~\cite{1form} because the unbounded, 
outside face, which is not in $G$, is assumed in this section to be closed and simple, 
and therefore would have null index. 
Note that the machinery from~\cite{1form} to deal with degenerate cases
isn't needed here because vertices, by definition, cannot coincide ($V$ is not a multiset). 
All  proofs in this section, except for that of Thm.~\ref{th:main}, are  in Appendix F. 

The one-forms defined above satisfy the following property. 

\begin{lemma}\label{lem:non-negative}
	Given a non-vanishing $\xi^n$, the sum of indices of interior vertices ($V\setminus W$) of $\bar{G}$ is non-negative. 
\end{lemma}

The next two lemmas 
relate the presence of edge foldovers and 
the ECE 
property of Definition~\ref{def:ece}  to the indices of vertices in $\bar{G}$. 

%

\begin{lemma}\label{lem:index-1}
If $\bar{G}$ has an edge foldover, then there is a non-vanishing one-form $\xi^n$ such
that $\mathbf{ind}_{_{\xi^n}}(v) < 0$ for some interior vertex $v\in V\setminus W$. 
\end{lemma}

\begin{lemma}\label{lem:index1}
Given $n\in\mathbb{S}^1$ and non-vanishing one-form $\xi^n$, if $\bar{G}$ has an interior vertex $v\in V\setminus W$ with index
$\mathbf{ind}_{_{\xi^n}}(v)=1$, then there is a face $f$ of
$G$ that does not satisfy the empty circum-ellipse 
property. 
\end{lemma}

The above provides the necessary tools to prove the following key lemma. 

\begin{lemma}\label{lem:ef}
$\bar{G}$ 
has no edge foldovers. 
\end{lemma}

Finally, the absence of edge foldovers, 
together with a simple and closed boundary, 
is sufficient to show that $\bar{G}$ is 
embedded.

\begin{lemma}\label{lem:main-weak}
If its (topological) boundary is simple and closed, then the straight-line dual $\bar{G}$ of an orphan-free diagram, with vertices incident to at most three sites, is an embedded triangulation. 
\end{lemma}

As shown in Appendix G, even if there are non-generic Voronoi vertices that are incident to more than three sites, the dual is composed of faces each of which is convex and satisfies the ECE condition (Def.~\ref{def:ece}). Each convex face can be triangulated (e.g.\ in a fan arrangement) in such a way that individual triangles satisfy the ECE condition with the same witness ellipse as the face from which they are derived. This leads to the following:

\begin{theorem}\label{th:main}
If its (topological) boundary is simple and closed, then the straight-line dual $\bar{G}$ of an orphan-free diagram is an  embedded polygonal mesh with convex faces. 
\end{theorem}

\section{Final Results}\label{sec:gen}

We can combine Corollary~\ref{col:simple-boundary} and Theorem~\ref{th:main} into
\begin{theorem}\label{gamma}
If the metric $Q$ has bounded ratio of eigenvalues, then the dual of an orphan-free Voronoi diagram with respect to $Q$ is an embedded polygonal mesh with convex faces, and covers the convex hull of the sites. 
\end{theorem}

This result can be generalized, dropping the bounded anisotropy condition on $Q$, but at the cost of losing the convex hull property, as shown in Appendix H:
\begin{theorem}\label{th:final}
The dual of an orphan-free Voronoi diagram is an embedded polygonal mesh with convex faces. 
\end{theorem}

From the above results, Corollary~\ref{cor:Pij}, and the definition of the dual, it follows (see Appendix H) that

\begin{corollary}\label{uniqueVD}
An orphan-free anisotropic Voronoi diagram is composed of unique (connected) vertices, edges, and faces.
\end{corollary}


%
Finally, we note that Thm.~\ref{th:main} can be combined with existing conditions for orphan-freedom~\cite{avd}, resulting in a simple and natural condition for a set of sites to induce an embedded polygonal mesh as the dual of their anisotropic Voronoi diagram:
\begin{corollary}\label{cor:enet}
	If $V$ is an asymmetric $\epsilon$-net w.r.t.\ $D$, $Q$ a continuous metric with metric variation $\sigma$, 
		and $\epsilon\sigma \le 0.09868$, then the dual of the anisotropic Voronoi diagram of $V$
		is an embedded polygonal mesh with convex faces.
\end{corollary}
where an asymmetric $\epsilon$-net is simply a weaker form of $\epsilon$-net defined on non-symmetric functions $D$, which can be computed with the iterative algorithm of~\cite{Gonz}, and the metric variation $\sigma$ is a Lipschitz-type condition  on $Q$~\cite{avd}. 
The above condition is known to be conservative, and there may be simpler conditions to achieve orphan-freedom. 
As a practical observation, Du and Wang~\cite{DW} report orphans to be a rare occurrence in their experiments. 

\section{Proof-of-concept Implementation}\label{sec:implementation}

\begin{figure}[ht]
\centering
\subfigure[]{
\includegraphics[height=2.6cm]{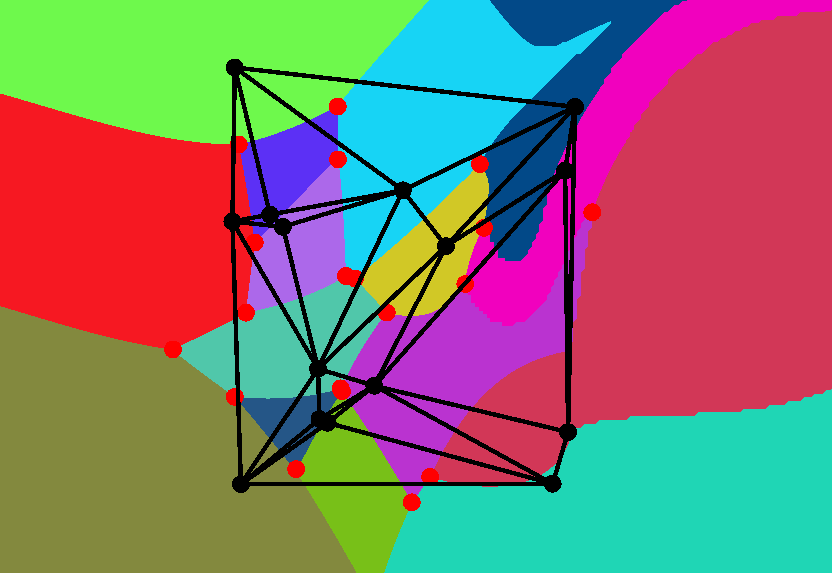}
\label{fig:img_a}
}
\subfigure[]{
\includegraphics[height=2.6cm]{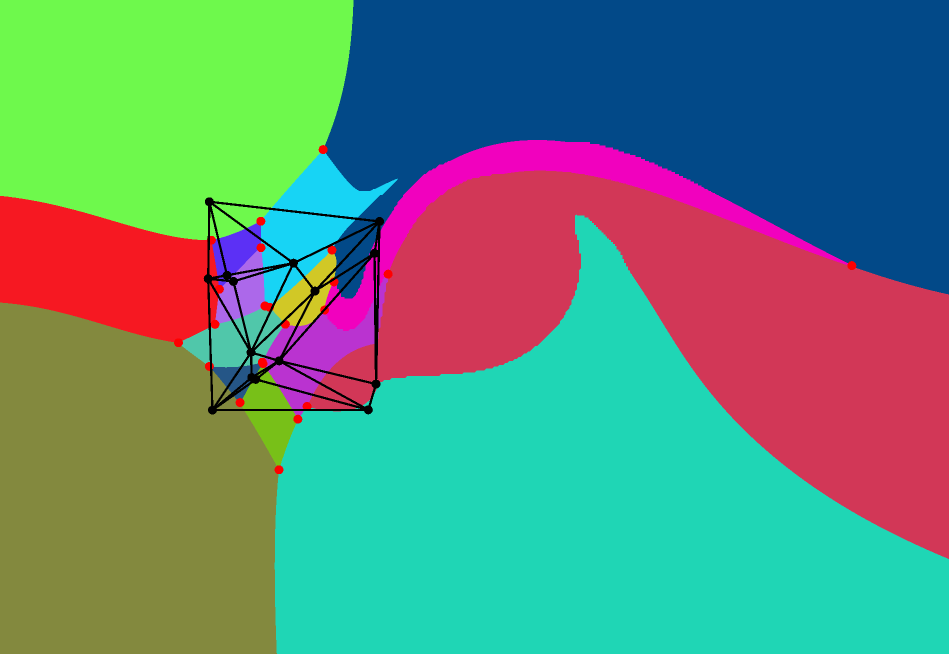}
\label{fig:img_b}
}
\subfigure[]{
\includegraphics[height=2.6cm]{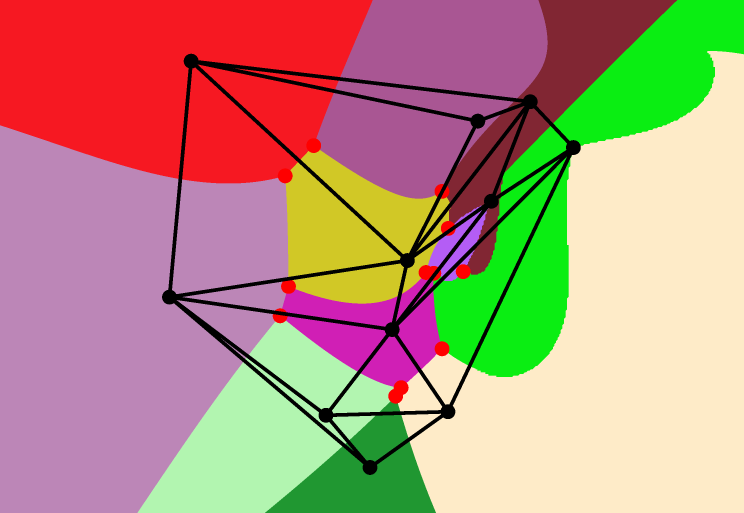}
\label{fig:img_c}
}
\subfigure[]{
\includegraphics[height=2.6cm]{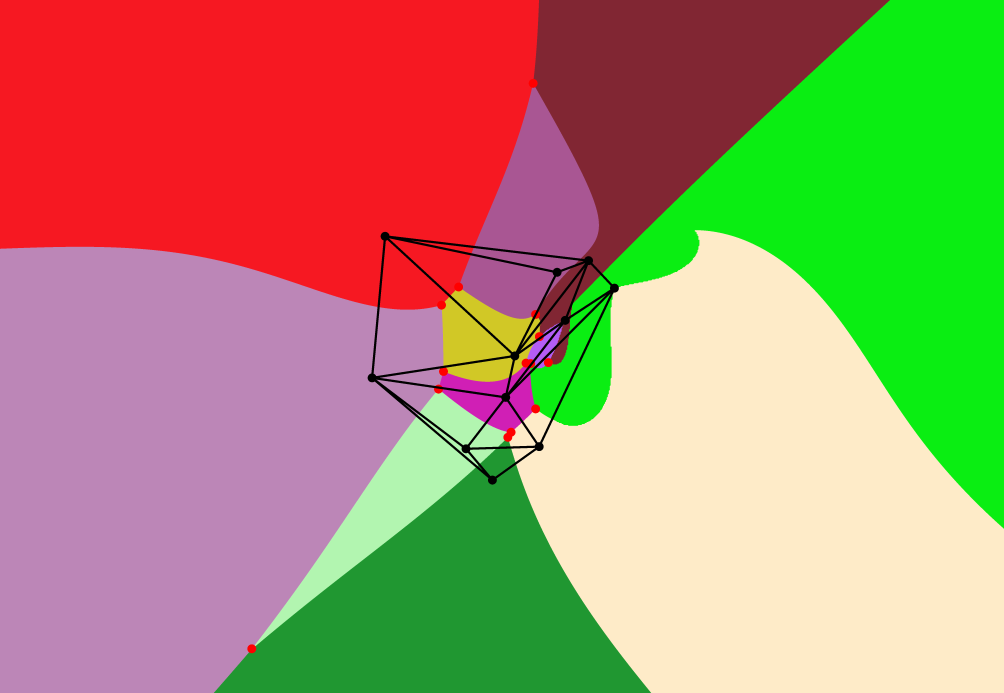}
\label{fig:img_d}
}\quad
\subfigure[]{
\includegraphics[height=2.6cm]{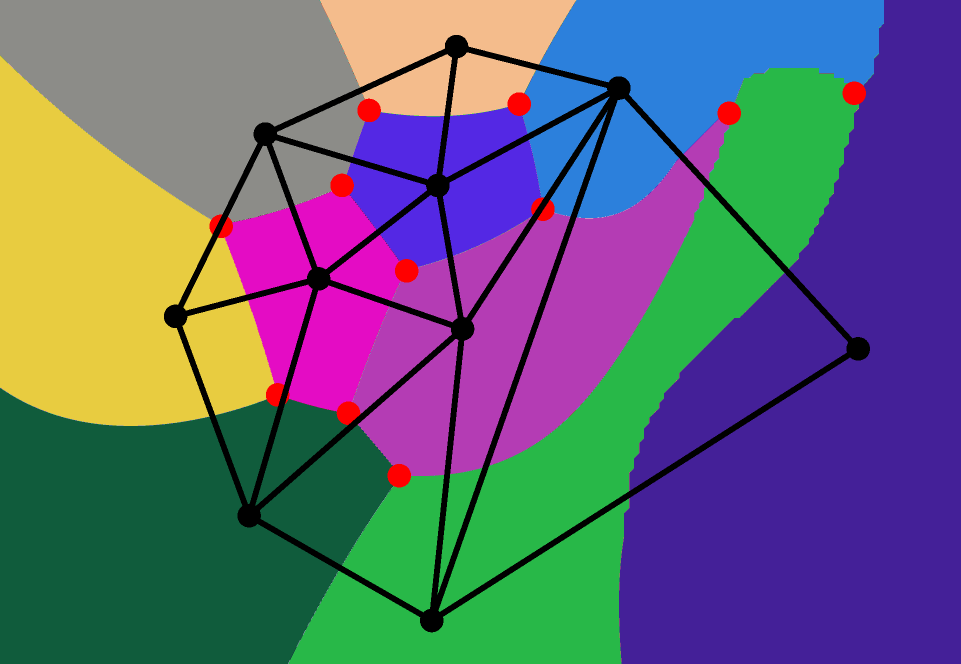}
\label{fig:img_e}
}
\subfigure[]{
\includegraphics[height=2.6cm]{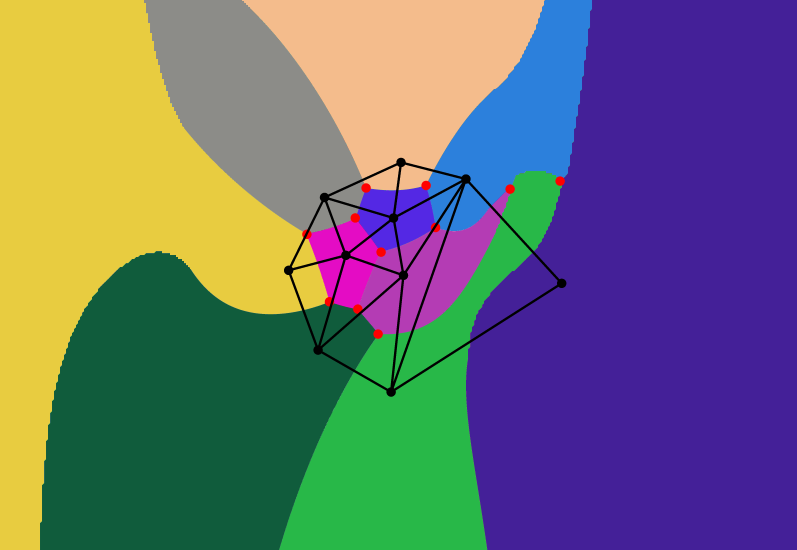}
\label{fig:img_f}
}
\subfigure[]{
\includegraphics[height=2.6cm]{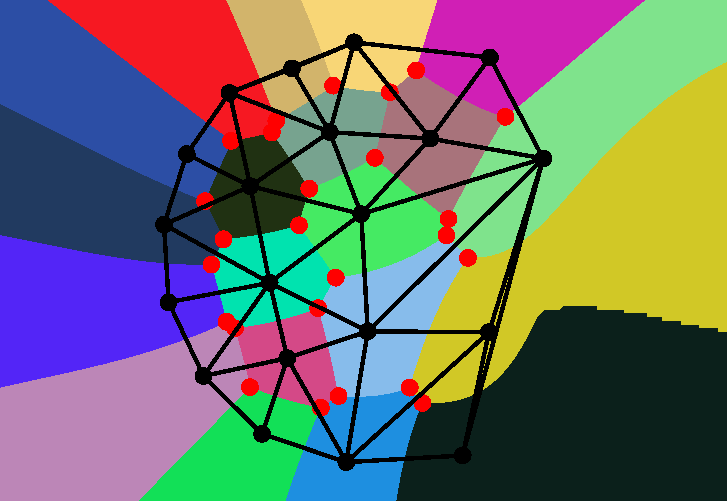}
\label{fig:img_g}
}
\subfigure[]{
\includegraphics[height=2.6cm]{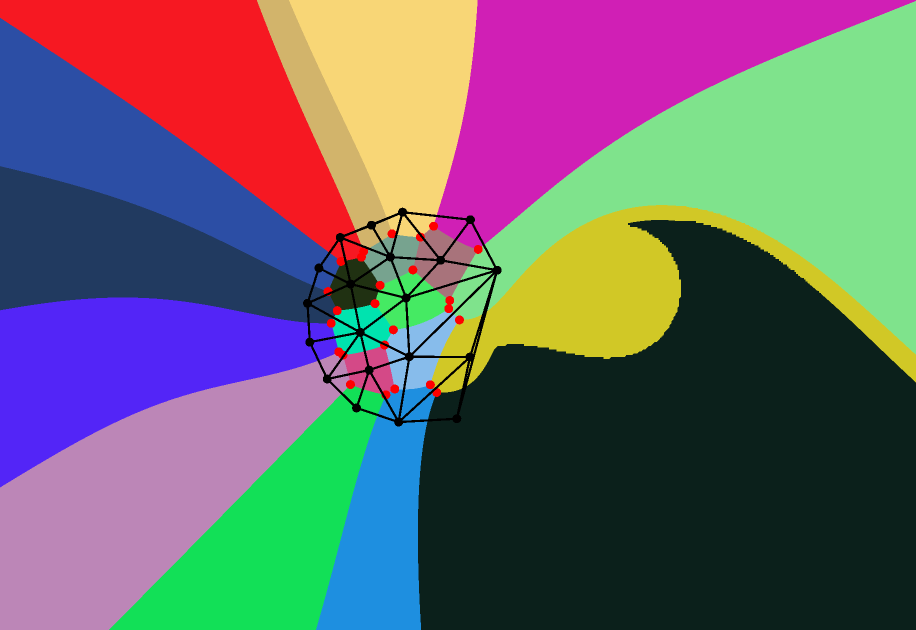}
\label{fig:img_h}
}
\caption{
Anisotropic Voronoi diagrams, and their duals generated by our
proof-of-concept implementation. 
Voronoi vertices are marked in red, while dual vertices (sites) and edges are drawn
in black.}
\end{figure}

Though not aiming for an efficient implementation, 
we implemented a simple proof-of-concept that constructs anisotropic Voronoi diagrams
like the ones considered in this paper, and their duals
(Fig. 3). 
A closed-form metric, which 
has bounded ratio of eigenvalues, is discretized on a
fine regular grid, and linearly interpolated inside grid elements, resulting in a
continuous metric. The sites are generated randomly (Figs.~\ref{fig:img_a}
and~\ref{fig:img_b}), or using a combination of random, and equispaced
points forming an asymmetric $\epsilon$-net (remaining figures). 

The primal diagram was obtained using front propagation from the sites
outwards, until fronts meet at Voronoi edges. 
The runtime is proportional to the grid size, since every grid-vertex is visited exactly six times (equal to their valence). 
The implementation does not 
guarantee the correctness of the diagram unless it \emph{is} orphan-free, and serves to verify the claims of the paper since well-behave-ness of the dual is predicated on that of the primal. 
%

The two main claims of the paper are clearly illustrated in these examples. 
In all examples, the dual covers the convex hull of the vertices
(Thm.~\ref{thm:boundary}), is a
single cover, embedded with straight edges without edge crossings
(Thm.~\ref{th:main}), 
and has no degenerate faces (Lem.~\ref{lem:degen}). 
By focusing on the primal diagrams (second and fourth column), further claims in
the paper become apparent, namely that Voronoi regions are simply connected (Lem.~\ref{lem:sc}), 
Voronoi vertices are unique (Cor.~\ref{uniqueVD}), 
Voronoi edges are connected (Cor.~\ref{cor:Pij}), 
unbounded regions correspond to boundary dual vertices, and unbounded
edges of the Voronoi diagram correspond to boundary dual edges.

\section{Conclusion and Open Questions}

We studied the properties of duals of orphan-free anisotropic Voronoi diagrams, for the
purposes of constructing triangulations on the plane. 
The main result (Theorems~\ref{th:main} and~\ref{thm:boundary}, Cor.~\ref{cor:enet}) is that
the dual, with straight edges and having the sites as vertices, is embedded
and covers the convex hull of the sites, mirroring similar results for
ordinary Voronoi diagrams and their duals. 

A few, somewhat less important  properties are proven, including the fact
that every primal region is simply connected, that elements of the primal are unique (Cor.~\ref{uniqueVD}, 
but perhaps most significantly
that every face in the dual satisfies an \emph{empty circum-ellipse}
property that has a direct parallel in the empty circum-circle property of
ordinary diagrams, and is the basis for proving that it is embedded with
straight edges. 

Perhaps the most important outstanding question may be whether these ideas
extend to higher dimensions. 
The results in Secs.~\ref{sec:boundary} and~\ref{sec:interior},
except for Lem.~\ref{lem:sc}, can be trivially extended to n dimensions. 
Sec.~\ref{sec:boundary} has been written only for the two-dimensional case,
but a similar construction, and the same argument would work in higher
dimensions (Lem.~\ref{lem:Sn} being a hint of this). 
It is the argument in Sec.~\ref{sec:interior} that becomes problematic. 
While the ECE property is shown to be sufficient to prevent foldovers in the
triangulation, it is not sufficient in 
higher dimensions. In particular, fixing the boundary to be simple and convex, 
there are simple arrangements of tetrahedra in $\mathbb{R}^3$ that contain 
face foldovers but do not break the ECE property. 
We plan to study these question next. 


%
%
%
%
%
%
%
%
%
%
%
%
%
%
%
%
%


\newpage
\bibliographystyle{plain}
\bibliography{vddw3}




\newpage
\section*{Appendix A}\label{app:prelim}

\begin{figure}[ht]
\centering
\subfigure[]{
\includegraphics[height=2.5cm]{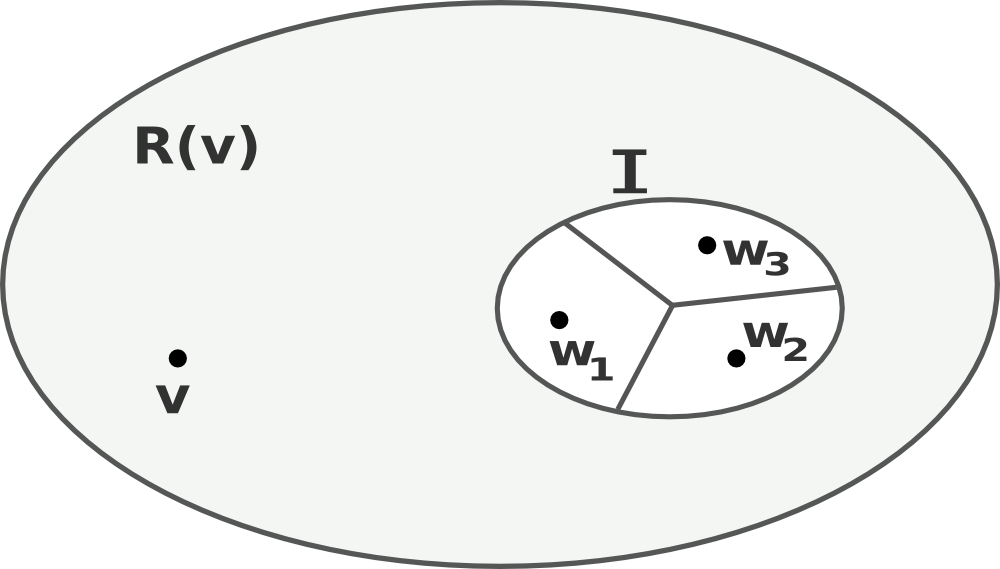}
\label{fig:sca}
}
\subfigure[]{
\includegraphics[height=2.5cm]{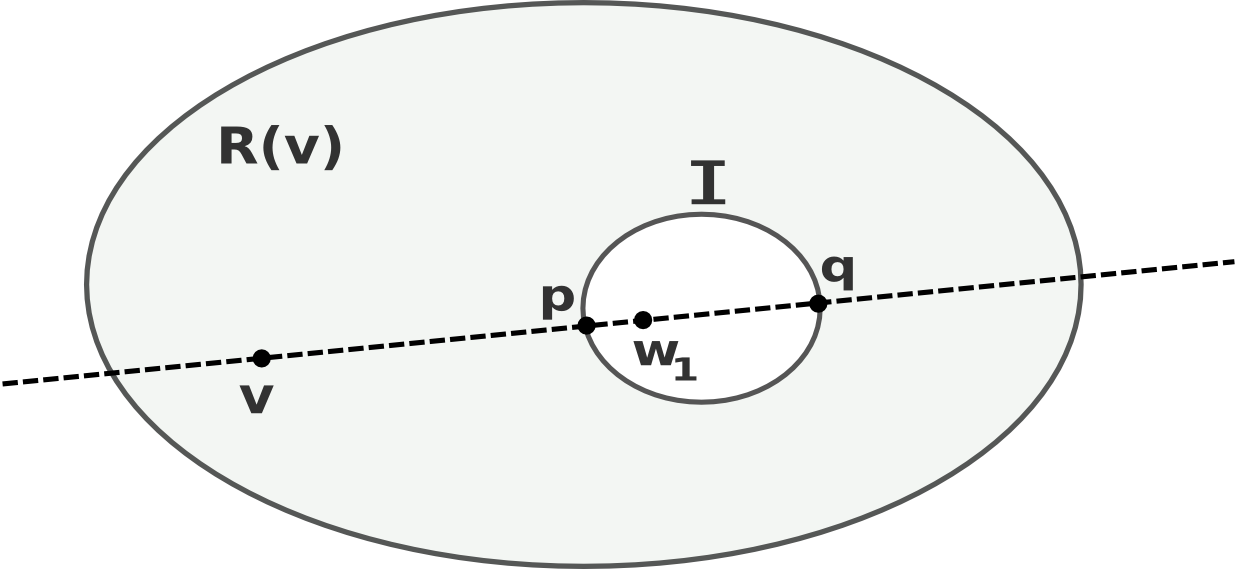}
\label{fig:scb}
}
\label{fig:sc}
\caption{Diagram for the proof of Lem.~\ref{lem:sc}.}
\end{figure}

\noindent{\bf Lemma~\ref{lem:midpoint}}
\emph{Given two sites $v,w\in V$, $v\ne w$, the only point equidistant to $v,w$ in their supporting line is the midpoint $(v+w)/2$. }
\begin{proof}
	If $p$ is in the supporting line of $v,w$, then $p=(1-\lambda)v+\lambda w$, $\lambda\in\mathbb{R}$. 
	$p$ is equidistant to $v,w$ (i.e.\ $D(v,p)=D(w,p)$) if and only if
		$\lambda^2 (v-w)^t Q_p (v-w) = (1-\lambda)^2 (v-w)^t Q_p (v-w)$.

	Since $Q$ is positive definite and $v\ne w$, it is $(v-w)^t Q_p (v-w) > 0$, 
	and so the above equality holds iff $\lambda=1/2$ ($p$ is the midpoint).
\end{proof}


\noindent {\bf Lemma~\ref{lem:sc}}
\emph{%
Every Voronoi region 
of an orphan-free anisotropic Voronoi 
diagram in 
$\mathbb{R}^2$ 
is simply connected. 
}
\begin{proof}	
A multiply connected region $R(v)\subset\mathbb{R}^2$ is path connected, and is
such that there is a map
$f:\mathbb{S}^1\rightarrow R(v)$ that cannot be continuously contracted to a point. We
can assume $f$ injective (simple), since a non-injective $f$ can always be broken up
into injective pieces, at least one of which must be such that it cannot be
continuously contracted to a point (or else $f$ would be). 
By the Jordan curve theorem, $f$ encloses a bounded set $B$. Since $f$
cannot be continuously contracted to a point, there must be
$I\subseteq B$ with $I\subset\mathbb{R}^2\setminus R(v)$. Since $I$ is a subset of $B$, 
and $B$ is bounded, $I$ is bounded. 



$I$ is part of the Voronoi diagram, so it must be composed of one or more
Voronoi cells. Because the diagram is orphan-free, each cell in the Voronoi
diagram contains its generating site: $I=\displaystyle{\cup_i R(w_i)}, i=1,\dots,m$
(Fig~\ref{fig:sca}), for some number $m>0$. 

We can now consider what the diagram would look like if we remove, one by
one, all of the sites $w_i$ in $I$, until only
one ($w_1$) is left.
The new set of sites is $V'=V\setminus \{w_i : i=2,\dots,m\}$. 
	From the definition of Voronoi diagram it is clear that 
	the region corresponding to site $w_1$ in the new diagram is $R'(w_1)\subseteq I$
(since no point is \emph{strictly} closer to $w_1$ than to $\{w_i : i=1,\dots,m\}$), and
$R'(w_1)\ne\phi$
(since, by Lem.~\ref{lem:interior}, it is $w_1\in
\mathbf{int\ }{R'(w_1)}$). 

In particular, since $R'(w_1)\ne\phi$, and $R'(w_1)$ is connected, 
$R'(v)$ is \emph{still} multiply-connected,
with $R'(w_1)$ playing the role of $I$ in the new diagram (Fig.~\ref{fig:scb}). 
We show that this cannot be the case, 
resulting in a contradiction. 

Since $I$ is bounded, then $R'(w_1)\subseteq I$ must be bounded as well. 
And since $w_1$ is an interior point of $R'(w_1)$, and $v$ is 
an interior point of $R'(v)$ (and therefore an interior point of the complement of $R'(w_1)$), 
then the line passing through $v$ and $w_1$ must intersect $R'(v)$ at
least at two distinct points $p$ and $q$. 

Since the boundary between $R'(v)$ and $R'(w_1)$ is composed of
equidistant points to $v,w_1$, then both $p,q$ must be equidistant to $v,w_1$. 
By Lem.~\ref{lem:midpoint}, only one point on the line connecting $v,w_1$ can be equidistant to them, a contradiction.

%
%
%
%
%

\end{proof}


\section*{Appendix B: Colinear sites}\label{app:colinear}

\begin{figure}[ht]
\centering
\subfigure
{
\includegraphics[height=3.5cm]{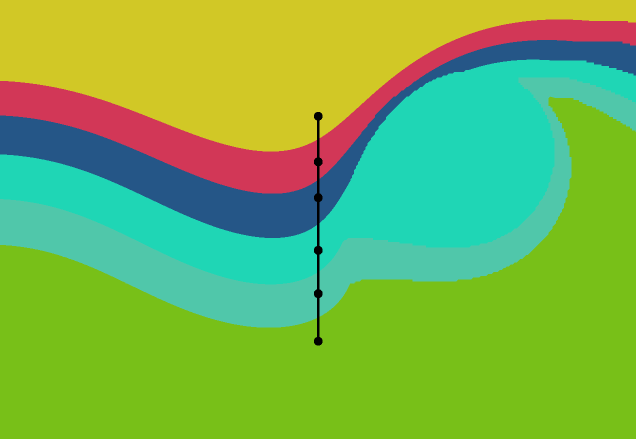}
}
\subfigure
{
\includegraphics[height=3.5cm]{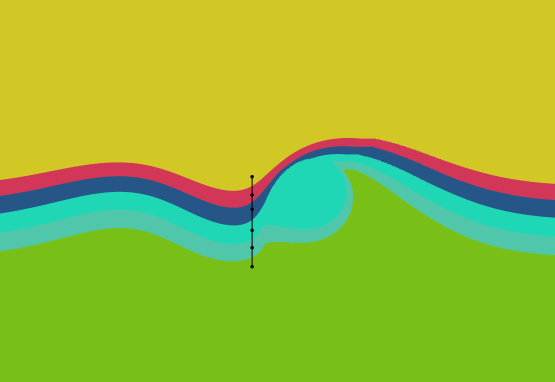}
}
\caption{If all sites are colinear, the dual is always a polygonal chain.}
\end{figure}

If all the sites are colinear then the structure of the Voronoi diagram
$\tilde{P}$ is
greatly simplified, and always has the form shown in
Fig.\  5.
In particular, $\tilde{P}$ has no vertices aside from $p_\infty$ since 
vertices are equidistant to three or more sites, 
and no point $p\in\mathbb{R}^2$ can be equidistant to three colinear
sites (since points equidistant to $p$ form an ellipse, and a line
intersects an ellipse at most twice). 

Consider the set of sites $V=\{w_l : l=1,\dots,m\}$ ordered linearly along
their supporting line. 
We shown that the graph dual to $P$, having $V$ as vertices,
has edges $\{(w_l,w_{i+1}) : l=1,\dots,m-1\}$. 

For every pair of sites $w_l,w_{l+1}$, because they are consecutive, and by Lem.~\ref{lem:midpoint}, 
the midpoint $m_{l,l+1} = (w_l+w_{l+1})/2$ is closest and equidistant to $w_l,w_{l+1}$,
and therefore the edge $(w_l,w_{l+1})$ is in the dual. 

We finally show that every dual edge $(w_i,w_j)$ is of the form
$(w_l,w_{l+1})$. Assume otherwise, and therefore that, since $w_i,w_j$ are
not consecutive, there is some $w_k$ between them. 
Because there is a dual edge $(w_i,w_j)$, the
corresponding primal edge $\tilde{E}_{w_i,w_j}$ in $\tilde{P}$ is not empty, and therefore there is some 
point $p_{ij}\in \tilde{E}_{w_i,w_j}$ that is closest and equidistant to $w_i,w_j$. 
By the convexity of $D(\cdot,p_{ij})$, and the fact that $w_k\in \overline{w_i
w_j}$, it is $D(w_k,p_{ij}) < D(w_i,p_{ij}) = D(w_j,p_{ij})$, a
contradiction. Therefore $(w_i,w_j)$ are consecutive sites. 

Since we have characterized the set of edges and vertices, and
there are no (interior) faces in the dual, this completely determines the
dual when all sites are colinear. 

Perhaps interestingly, this means that, in the colinear case, the structure of the dual does not depend on the metric.

\setcounter{section}{6}
\section*{Appendix C: Technical Lemmas}\label{app:technical}

We include in this appendix all claims that, 
although needed to prove the results in this paper, 
are used only as intermediate steps. 
Some of them reveal useful aspects of the structure of the problem. 
In particular, Lem.~\ref{lem:halfspace} is used as a basic step throughout the
paper. 
The proofs are quite technical, and can be skipped on a first read without
affecting the comprehension of the remainder of the paper.

The following result 
states that, even though the regions and edges of an 
anisotropic Voronoi diagram may be unbounded, there is a sufficiently
large ball outside which there are no vertices. 

\begin{lemma}\label{lem:bounded-vertices}
The set of vertices in the primal diagram $\tilde{P}$ is bounded. 
\end{lemma}
\begin{proof}
Let $M_p = R_p \Lambda_p^{1/2} R_p^t$ be the unique symmetric, positive definite square root of $Q_p$. 
Define $Q'_p = Q_p\cdot\lambda^{-1}_1(p) = R_p \diag\left[\lambda_2(p)/\lambda_1(p), 1\right] R_p^t$ 
(which is continuous by virtue of the continuity of $Q$ and $\lambda_1(p)$) 
and its symmetric positive definite square root $M'_p$. 

Pick three sites $v_i,v_j,v_k\in V$. If they are colinear then no point is 
equidistant to them, and so the lemma holds vacuously.  
If they are not colinear, 
consider some $Q'\in\mathbb{R}^{2\times 2}$ symmetric,
positive definite with ratio of eigenvalues in the range $[1,\gamma^2]$. 
Define  $\Psi(a,b,c)$ to be the {unique} center
of a circle passing through three (non-colinear) points $a,b,c$. $\Psi$ is
continuous wherever $a,b,c$ are not colinear. 

If $p$ is equidistant to $v_i,v_j,v_k$, and has $Q'_p=Q'$, 
then clearly 
\[(v_i-p)^t Q' (v_i-p) = (v_j-p)^t Q' (v_j-p) = (v_k-p)^t Q' (v_k-p) \]
or equivalently $\|M' v_i- M' p\| = \|M'v_j - M' p\| = \|M' v_k- M' p\|$. 
This means that $M' p = \Psi(M' v_i, M' v_j, M' v_k)$, 
and so $p = M'^{-1} \Psi(M' v_i, M' v_j, M' v_k)$ is the unique point
that can be equidistant to $v_i,v_j,v_k$ while having $Q'_p=Q'$.

Since the space of $Q'$ is compact and $M'^{-1} \Psi(M' v_i, M' v_j, M'
v_k)$ is continuous w.r.t.\ $M'$, then set of possible equidistant points to
$v_i,v_j,v_k$ is compact, and therefore bounded. 
Since this is true for every triple $v_i,v_j,v_k\in V$, and there is a
finite number of triples, the set of possible equidistant points to three or
more sites is compact, and therefore bounded. 
Since, by definition, vertices are equidistant to three or more sites, the lemma follows. 
\end{proof}

For the sake of conciseness, we define: 
\begin{definition}
	The open ellipse centered at $p$ with $m$ in its boundary is 
	\[\theta_p(m) = \{x\in\mathbb{R}^2 : D(x,p) < D(m,p)\} =
\{x\in\mathbb{R}^2 : (x-p)^t Q_p (x-p) < (m-p)^t Q_p (m-p)\}\]
\end{definition}

Given some $p\in\mathbb{R}^2$, we make frequent use of a map that transforms points
$q\in\mathbb{R}^2$ into $q'=M'_p q$, where 
$M'_p=R_p \diag[(\lambda_2(p)/\lambda_1(p))^{1/2}, 1]R_p^t$ defined in the
proof of Lem.~\ref{lem:bounded-vertices}. 
We follow the convention of denoting by $q'$ the transformed version of point
$q$, where it will be clear from the context which matrix $M'_p$ is used in
the transformation. 
Since the eigenvalues of $M'_p$ are in the range $[1,\gamma]$, it's easy to
see that: 
\begin{itemize}
\item \emph{Given $p,q\in\mathbb{R}^2$, it is $\|p-q\| \le \|p'-q'\| \le \gamma \|p-q\|$}
\item \emph{If $d(p,l)>0$ is the minimum Euclidean distance between a point
$p$ and a line $l$, then $\gamma^{-1} \le d(p',l')/d(p,l) \le \gamma$.}\\ \\
\emph{Proof:} Given $a,b\in l$, $a\neq b$, it is $d(p',l') = |(p'-b')\times (a'-b')| /
\|p'-b'\|\cdot\|a'-b'\|$ where the numerator is $|(p'-b')\times (a'-b')| =
\gamma |(p-b)\times (a-b)|$, and the denominator satisfies 
$\|p-b\|\cdot\|a-b\| \le \|p'-b'\|\cdot\|a'-b'\| \le \gamma^2
\|p-b\|\cdot\|a-b\|$. Therefore $d(p',l')/d(p,l) \in [\gamma^{-1},\gamma]$. 
\end{itemize}
The following technical lemma is 
the basis for the proof of Lem.~\ref{lem:halfspace}. 

\begin{figure}[ht]
\centering
\subfigure[]{
\includegraphics[height=3cm]{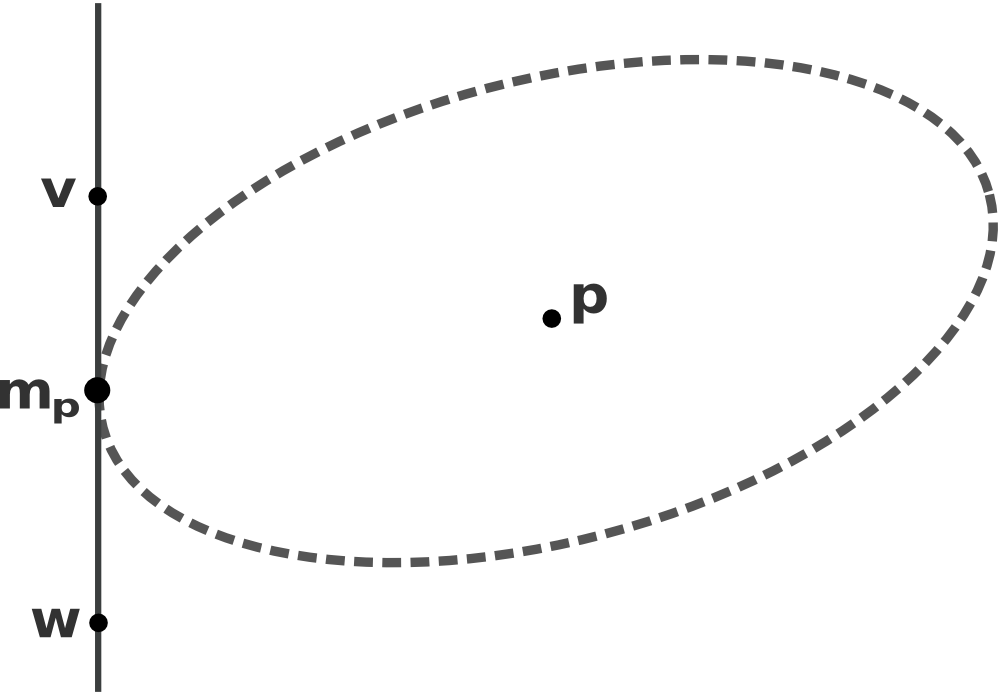}
\label{fig:thetaa}
}
\subfigure[]{
\includegraphics[height=3cm]{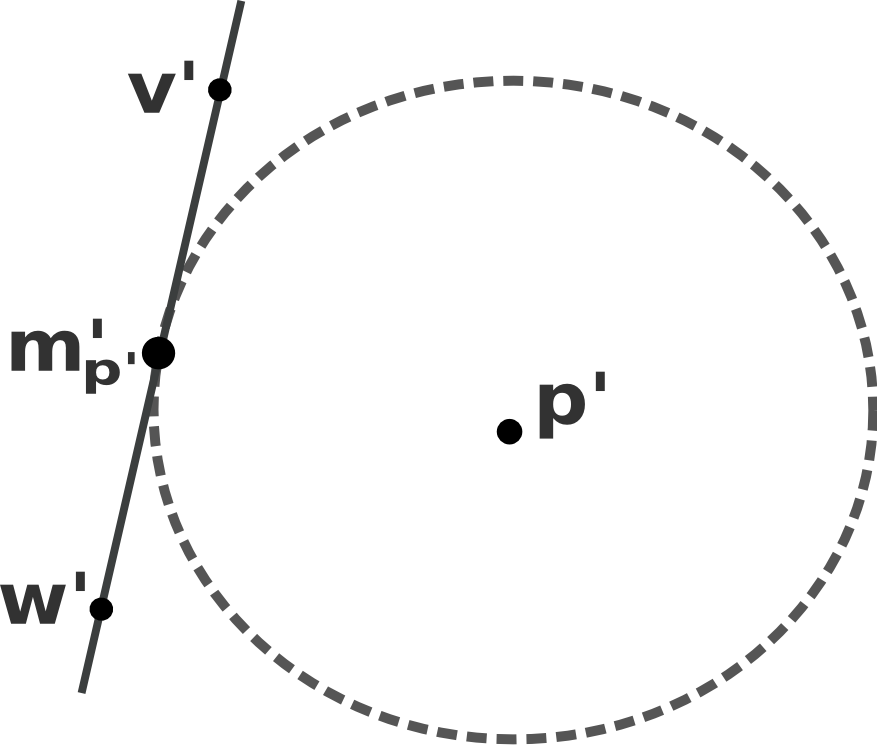}
\label{fig:thetab}
}
\label{fig:theta}
\caption{Diagram for the proof of Lem.~\ref{lem:tech}.}
\end{figure}

\begin{lemma}\label{lem:tech}
	Given $v,w\in\mathbb{R}^2$, with supporting line $l$, and 
an open half space $H^{+}$ on one side of $l$, 
as well as a set $S\subset H^{+}$ of points whose closest point in $l$
belongs to $\overline{v w}$, 
then 
for all $q\in H^{+}$, if $p\in S$ has
\[ \|p\| > \max\{\|v\|,\|w\|\} + \frac{\gamma^2 \displaystyle{\max_{m\in\overline{v w}} \|q-m\|} }{ 2 d(q,l) }\]
then 
$p$ is strictly closer to $q$ than to $l$. 
\end{lemma}
\begin{proof}
Given a point $p$ that satisfies the above constraint, 
by the definition of $S$, 
it is $m_p = \underset{m\in l}{\mathbf{argmin\ } } D(m,p)\in\overline{v w}$. 
Also, since it is $q\in H^{+}$ and since $H^{+}$ doesn't include $l$, 
$q\notin l$ and therefore $d(q,l) > 0$.
Because $m_p = \underset{m\in l}{\mathbf{argmin\ }} D(m,p)\in\overline{v w}$, 
it can be easily verified that the ellipse $\theta_p(m_p)$ must be 
tangent to $l$ at $m_p$ (Fig.~\ref{fig:thetaa}). 
When transforming it by $M'_p$, $\theta_p(m_p)$ becomes the circle
$\theta'_{p'}(m'_p)=\{x'\in\mathbb{R}^2 : \|x'-p'\| < \|m'_p-p'\|\}$
(Fig.~\ref{fig:thetab}), which is tangent to $l'$ at $m'_p$. 
The lemma then reduces to showing that $q'\in\theta'_{p'}(m'_p)$ since this
implies $q\in\theta_p(m_p)$ and this in turn means that $D(q,p) < D(m_p,p) = \min_{m\in l}
D(m,p)$, and therefore $q$ is closer to $p$ than to $l$. 


A simple calculation shows that $q'\in\theta'_{p'}(m'_p)$ iff 
$\|p'-m'_p\| > \|q'-m'_p\| / \left[2 d(q',l')\right]$. 
Since $m_p\in\overline{v w}$ implies $\|m_p\| \ge
\max\{\|v\|,\|w\|\}$, it is 
\begin{eqnarray*}
	\|p'-m'_p\| &\ge& \|p-m_p\| \ge \|p\| - \|m_p\| > \\
 \frac{\gamma^2 \displaystyle{\max_{m\in\overline{v w}} \|q-m\|} }{ 2 d(q,l) }
&\ge& \gamma^2 \frac{\|q-m_p\|}{ 2 d(q,l) } 
 \ge \frac{\|q'-m'_p\|}{2 d(q',l')}
\end{eqnarray*}
and so $q'\in\theta'_{p'}(m'_p)$, as claimed. 	
\end{proof}

Given $v_i,v_j\in V$ whose Voronoi regions are neighbors in $\tilde{P}$, 
define $H^{+}_{ij}$ and $H^{-}_{ij}$ to be the two open half spaces 
on either side of the supporting line
$l_{ij}$ of $v_i,v_j$. Clearly, $\{H^{+}_{ij}, H^{-}_{ij}, l_{ij}\}$ form a partition of
$\mathbb{R}^2$. 
If the edge $\tilde{E}_{v_i,v_j}$ of $\tilde{P}$ corresponding to sites $v_i,v_j$ is unbounded
then, it must be either $\tilde{E}_{v_i,v_j}\cap l_{ij}$, $\tilde{E}_{v_i,v_j}\cap H^{+}_{ij}$, or
$\tilde{E}_{v_i,v_j}\cap H^{-}_{ij}$ unbounded.
By Lem.~\ref{lem:midpoint}, it is $\tilde{E}_{v_i,v_j}\cap l_{ij}\subseteq \{(v_i+v_j)/2\}$ 
and thus $\tilde{E}_{v_i,v_j}\cap l_{ij}$ is always bounded. 
Therefore it must be that either $\tilde{E}_{v_i,v_j}\cap H^{+}_{ij}$ or
$\tilde{E}_{v_i,v_j}\cap H^{-}_{ij}$ are unbounded (or both). 
The following lemma proves that to each unbounded Voronoi edge corresponds
at least one open half-space that does not contain any site. 
%


\begin{lemma}\label{lem:halfspace}
    Given an edge $\tilde{E}_{v_i,v_j}$ of the Voronoi diagram $\tilde{P}$, corresponding to neighboring sites $v_i,v_j\in
V$, 
if $\tilde{E}_{v_i,v_j}\cap  H^{+}_{ij}$ ($\tilde{E}_{v_i,v_j}\cap  H^{-}_{ij}$) is unbounded, then it is $ H^{+}_{ij}\cap V=\phi$
 ($ H^{-}_{ij}\cap V=\phi$), 
where $H^{+}_{ij},H^{-}_{ij}$ are open half spaces on either side of the
supporting line of $v_i,v_j$. 
\end{lemma}
\begin{proof}
Assume $w\in H^{+}_{ij}\cap V$. 
We make use of Lem.~\ref{lem:tech}, with $v_i,v_j,l_{ij},\tilde{E}_{v_i,v_j}\cap H^{+}_{ij}$
taking the role of $v,w,l,S$, respectively. 
Since $\tilde{E}_{v_i,v_j}\cap H^{+}_{ij}$ is unbounded, 
we can choose $p\in \tilde{E}_{v_i,v_j}\cap H^{+}_{ij}$ of sufficiently large norm, according to
Lem.~\ref{lem:tech}, 
such that $D(w,p) < D((w_i+w_j)/2,p)$ which, 
by the convexity of $D(\cdot,p)$, means that $D(w,p) < D((v_i+v_j)/2,p) <
D(v_i,p)=D(v_j,p)$. 
Since $p$ is closer to $w$ than to $v_i,v_j$, it is $p\notin \tilde{E}_{v_i,v_j}$, a
contradiction. 
\end{proof}

Lastly, we show that every point of sufficiently large norm is strictly closer to
the sites $W=V\cap\partial\mathcal{CH}(V)$ in the boundary of the convex
hull of $V$, 
than to sites ($V\setminus W$) in its interior. 
This means that, outside of a sufficiently large ball, the Voronoi 
diagram of $V$ is the same as that of $W$, a fact critical in proving the
results of Sec.~\ref{sec:boundary}. The proof of this property makes use of the
bound in the ratio of eigenvalues of $Q$, and it can easily be shown that
counter-examples exist in cases where no such bound exists. 
The proof assumes that the origin of $\mathbb{R}^2$ is chosen to be in the interior of $\mathcal{CH}(V)$.

\begin{figure}
\centering
\subfigure[]{
\includegraphics[height=2.5cm]{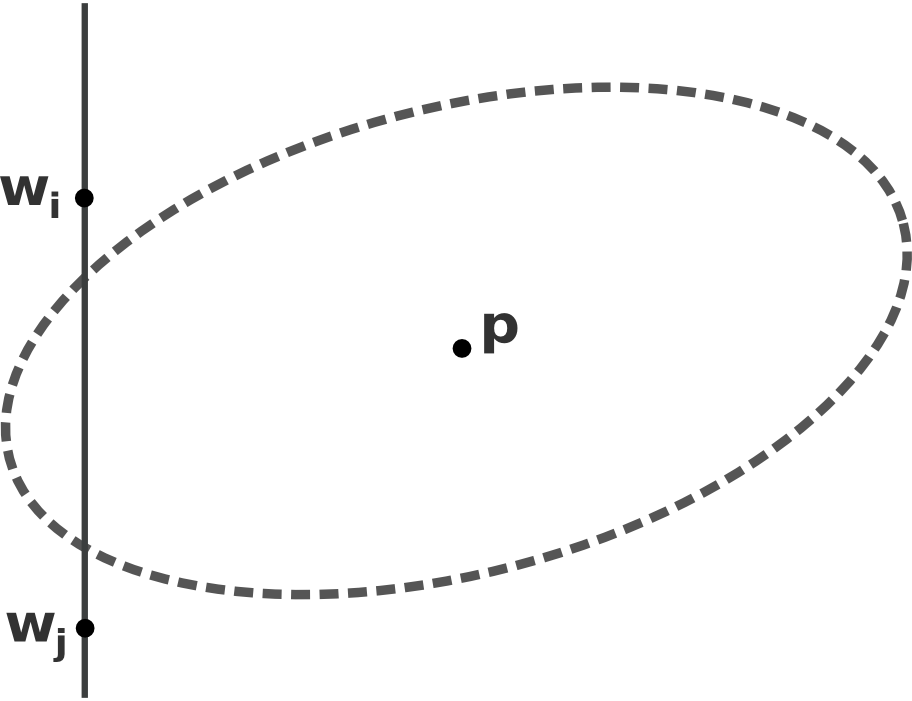}
\label{fig:wijcap}
}
\subfigure[]{
\includegraphics[height=2.5cm]{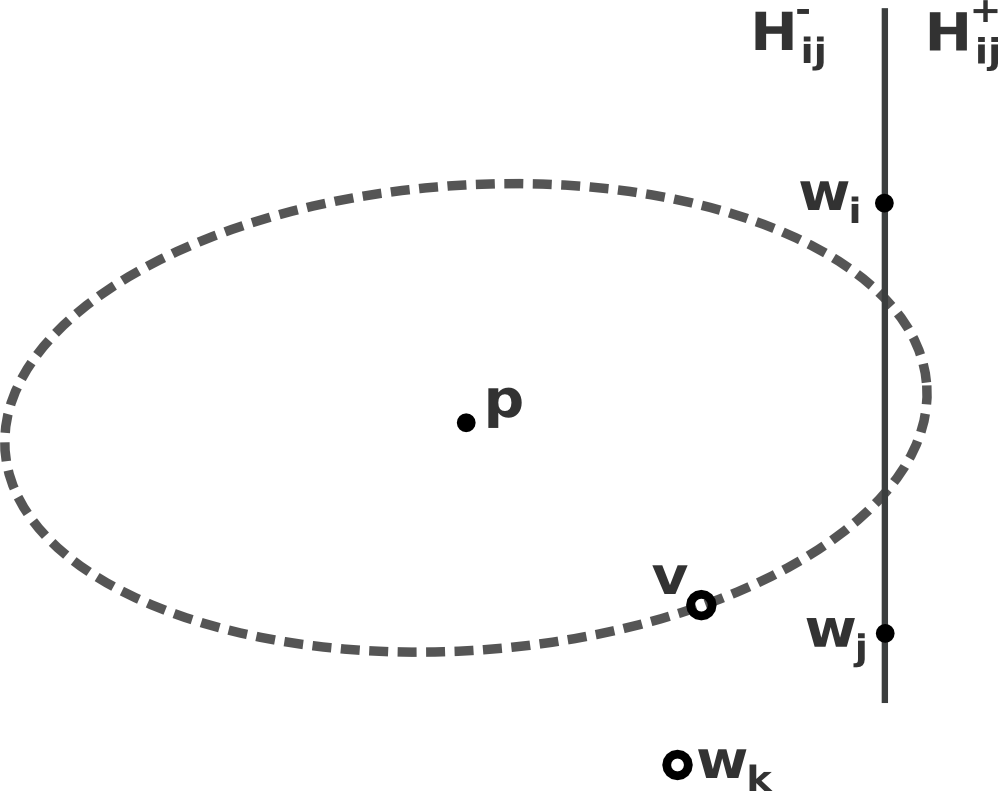}
\label{fig:vH-}
}
\quad\quad
\subfigure[]{
\includegraphics[height=2.5cm]{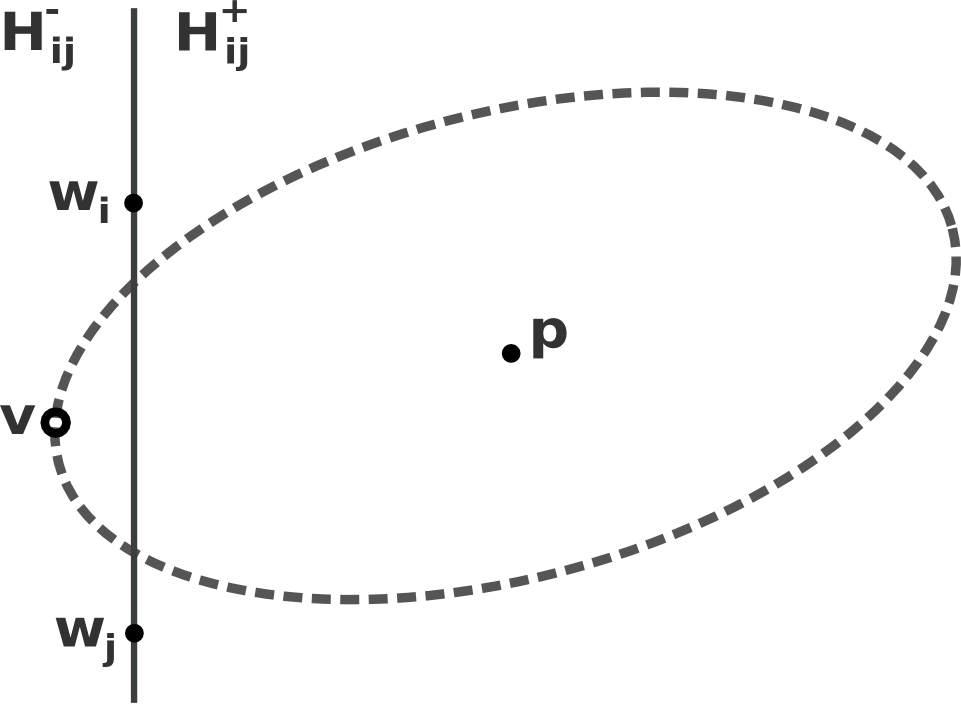}
\label{fig:vH+}
}
\label{fig:vH}
\caption{Diagram for the proof of Lem.~\ref{lem:VW}.}
\end{figure}

\begin{lemma}\label{lem:VW}
	There is $\rho$ such that all $p\in\mathbb{R}^2$ with $\|p\| > \rho$ 
are closer to $W$ than to $V\setminus W$. 
\end{lemma}
\begin{proof}
Begin by choosing $\rho \ge \max_{v\in V} \|v\|$ so that, since the
origin is in $\mathbf{int\ }{\mathcal{CH}(v)}$, every point of norm greater than $\rho$ is outside $\mathcal{CH}(V)$.
Pick $\rho$ such that it is also:
\begin{equation}\label{eqi}
\rho \ge \displaystyle{\max_{w\in W}\|w\| + \max_{(w_r,w_s)\in\mathcal{B},w_t\in W} \frac{\gamma^2 \max_{x\in\overline{w_r w_s}} \|x - w_t\|}{2 d(w_t, l_{ij})}}
\end{equation}
\begin{equation}\label{eqii}
\rho \ge \displaystyle{\max_{v\in V\setminus W, (w_i,w_j)\in\mathcal{B}} \|v\| + 
	\frac{\gamma^3 \|w_i - w_j\|^2 / 4 - d(v,l_{ij})^2\gamma^{-1}}{ 2 d(v,l_{ij})} } 
\end{equation}
The proof is by contradiction. 

If $p\notin\mathcal{CH}(V)$ is closest to
$v\in\mathbf{int\ }{\mathcal{CH}(V)}$ then, since $p\notin\mathcal{CH}(V)$, the ellipse 
$\theta_p(v)$ must intersect $\partial\mathcal{CH}(V)$. 
Because $p$ is not closer to any site in
$\partial\mathcal{CH}(V)$ than to $v$, then $\theta_p(v)$ cannot contain any
$w\in W$, and therefore $\theta_p(v)$ must intersect some segment $(w_i,w_j)\in\mathcal{B}$
as in Fig.~\ref{fig:wijcap} (to see this note that the intersection between an ellipse
and a line is an open line segment, and that the only way in which this
segment can overlap $\overline{w_i w_j}$ but not contain either $w_i,w_j$ is for
it to be $(\theta_p(v)\cap \overline{w_i w_j})\subseteq\overline{w_i w_j}$).

Consider the two open half spaces $H^{+}_{ij}$ and $H^{-}_{ij}$ on either side of
the supporting line of a boundary segment $(w_i,w_j)\in\mathcal{B}$, where
$H^{+}_{ij}$ is chosen to be on the ``empty" side
($H^{+}_{ij}\cap\mathcal{CH}(V)=\phi$, by Lem.~\ref{lem:halfspace}). 
Clearly, $\{H^{+}_{ij}, H^{-}_{ij}, l_{ij}\}$ partition $\mathbb{R}^2$. 
Since $\theta_p(v)$ must intersect some boundary segment, assume w.l.o.g.\   that
it intersects $(w_i,w_j)$. We show that $p$ cannot be in either
or the three sets in the partition, a contradiction.

Clearly, if $p\in l_{ij}$, then the fact that $\theta_p(v)\cap l_{ij}\subset\overline{w_i w_j}$
implies $p\in\overline{w_i w_j}$, which is precluded by
the fact that $p$ is outside $\mathcal{CH}(V)$. 

Assume $p\in  H^{-}_{ij}$, as in Fig.~\ref{fig:vH-}. 
Since not all sites are colinear, then there is $w_k\in W$ that is not colinear
with $w_i,w_j$, and therefore $w_k\in H^{-}_{ij}$. 
Eq.~\ref{eqi} has been constructed so that, for any choice of
$w_i,w_j,w_k\in W$, Lem.~\ref{lem:tech} implies that,
since $\|p\| > \rho$, if $m\in\overline{w_i w_j}$ is the
closest point to $p$ in $l_{ij}$, then it is $D(w_k,p) < D(m,p)$. 
The convexity of $D(\cdot,p)$, and the fact that
$\theta_p(v)\cap l_{ij}\subset\overline{w_i w_j}$, 
implies $D(w_k,p) <  D(m,p) < D(v,p)$, 
contradicting the fact that the site closest to $p$ is $v$.

Lastly, if we assume that $p\in  H^{+}_{ij}$ (Fig.~\ref{fig:vH+}), 
we obtain a similar contradiction using Eq.~\ref{eqii}. 
The ellipse $\theta_p(v)$ is transformed by $M_p$ into the circle
$\theta'_{p'}(v')$. 
Since it is $\theta_p(v)\cap l_{ij}\subset \overline{w_i w_j}$, 
then it is  $\theta'_{p'}(v')\cap l'_{ij}\subset\overline{w'_i w'_j}$. 
A simple calculation reveals that this can only be true if the radius of
$\theta'_{p'}(v')$ satisfies
\[ \| p' - v'\| \le \frac{ \|w'_i-w'_j\|^2/4 - d(v',l'_{ij})^2 
	}{ 2 d(v',l'_{ij}) } \]
However, from Eq.~\ref{eqii} and the fact that $\|p\| > \rho$, it is 
\begin{eqnarray*}
	\| p' - v'\| &\ge& \|p - v\| \ge \|p\| - \|v\| 
> \frac{\gamma^3 \|w_i - w_j\|^2 / 4 - d(v,l_{ij})^2\gamma^{-1}}{ 2 d(v,l_{ij}) } \\
&\ge& \frac{ \|w'_i-w'_j\|^2/4 - d(v',l'_{ij})^2 }{ 2 d(v',l'_{ij}) }
\end{eqnarray*}
a contradiction. 

Since we have shown that $p\notin l_{ij}\cup  H^{-}_{ij}\cup H^{+}_{ij}$, 
where $\{H^{+}_{ij}, H^{-}_{ij}, l_{ij}\}$ partition $\mathbb{R}^2$, 
this creates the contradiction that we were after. 
Therefore, given the above choice of $\rho$, 
every $p$ with $\|p\| > \rho$ is closer to the boundary sites $W$ than to any
site $V\setminus W$ in the interior of $\mathcal{CH}(V)$. 
\end{proof}


\section*{Appendix D}\label{app:simple}

\begin{figure}[ht]
\centering
\subfigure[]{
\includegraphics[height=3.5cm]{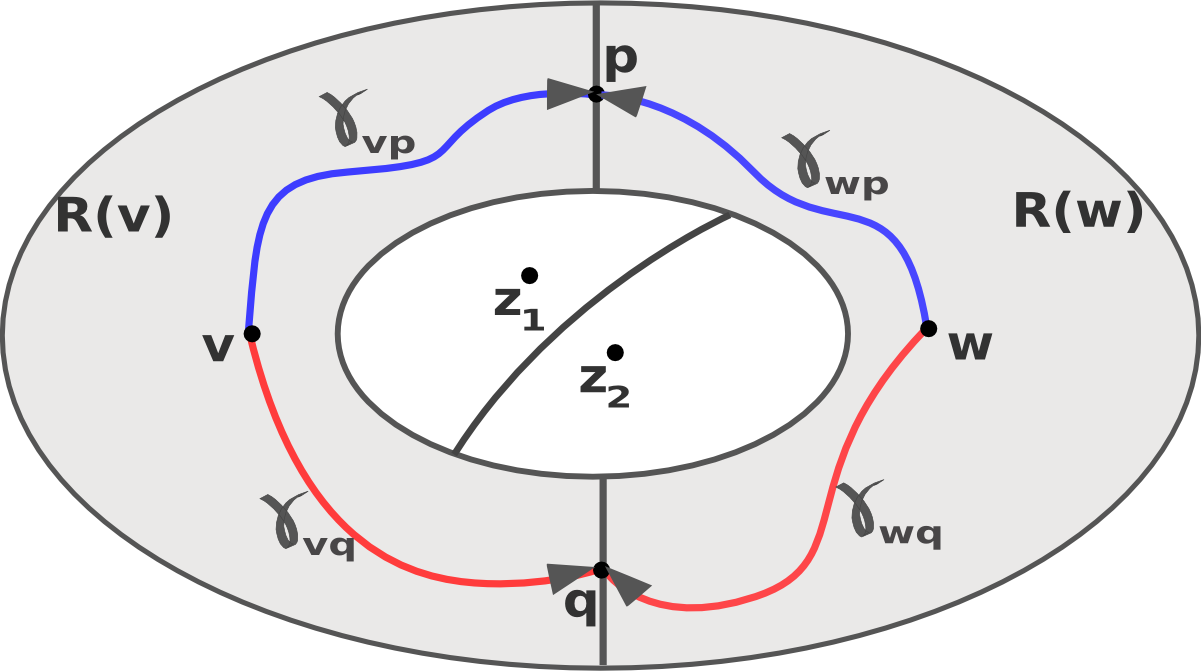}
\label{fig:3a}
}
\subfigure[]{
\includegraphics[height=3.5cm]{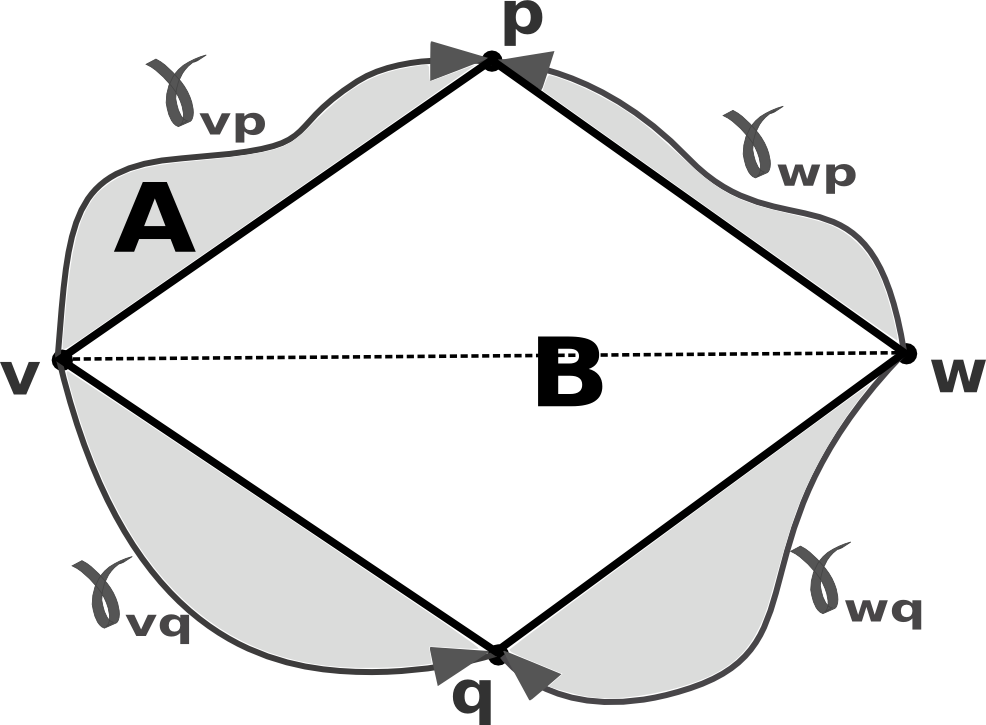}
\label{fig:3b}
}
\label{fig:subfigureExample}
\caption{Diagram for the proof of Lem.~\ref{lem:simple}.}
\end{figure}

\noindent{\bf Lemma~\ref{lem:simple}}
\emph{
$\tilde{G}$ and $\tilde{P}$ are simple (have no multi-edges or self-loops).
}
\begin{proof}

By definition, 
$\tilde{G}$ doesn't have self-loops since primal edges in $\tilde{P}$ 
are connected sets of equidistant points to distinct sites $v,w$, and
therefore they always connect distinct vertices. 
The only possible exception are self-loops in $\tilde{P}$ connecting
$p_\infty$ to itself. 
Because an edge of $\tilde{P}$ is a connected set $\tilde{E}_{v,w}$ of points
equidistant to two sites $v,w$, if the edge is a self-loop of $p_\infty$
then $\tilde{E}_{v,w}$ is unbounded on the two open half-spaces on either side of the
supporting line $l_{v,w}$ of $v,w$ and therefore, by Lem.~\ref{lem:halfspace},
neither half-space can contain any site. Therefore all sites are contained
in $l_{v,w}$, and are thus all colinear, a contradiction. Therefore there are no
self-loops. 

We prove by contradiction that $\tilde{G}$ has no multi-edges, 
and therefore, by duality, that $\tilde{P}$ doesn't either. 

	Consider two edges in the embedding $\tilde{G}$ that connect the same vertices $v,w$
(Fig.~\ref{fig:3a}).
One such edge is shown in blue in Fig.~\ref{fig:3a} and consists of a point $p\in R(v)\cap R(w)$, 
a path $\gamma_{vp}\subset R(v)$ from $v$ to $p$, 
and a path $\gamma_{wp}\subset R(w)$ from $w$ to $p$.
The second edge, shown in red, has analogous structure. 

From the construction of $\tilde{G}$, it is clear that the two edges must not intersect (other
than meeting at the endpoints), and therefore their union is a simple, closed
curve whose complement has, by the Jordan curve theorem, an interior $I$. 
$I$ must contain some regions corresponding to sites other than $v,w$, say $z_1,\dots,z_m$, 
or else the two edges would be connected and thus not counted as separate edges. 

Since the diagram is orphan-free, the regions inside $I$ contain their
generating sites and thus there are sites $z_1,\dots,z_m$ in the
interior of $I$. 
We show that this is not possible, resulting in a contradiction. 

Consider any site $z$ in $\mathbf{int\ }{I}$, as in Fig.~\ref{fig:3b}. 
We first split $\mathbf{int\ }{I}$, as in the figure, into two regions $A$ and
$B$, and show that $z$ cannot belong to either. 
Although $A$ and $B$ may not always be disjoint, it is easy to show that
they cover $\mathbf{int\ }{I}$ (i.e.\  $\mathbf{int\ }{I}\subseteq A\cup B$).

$A$ (shaded in Fig.~\ref{fig:3b}) is the set of points in $I$ that are \emph{inside} some segment
$\overline{vr}$ with $r\in\gamma_{vp}\cup\gamma_{vq}$, or inside
some segment $\overline{wr}$ with $r\in\gamma_{wp}\cup\gamma_{wq}$. 
If a site $z$ belongs to $A$ then, without loss of generality, there is $r_0\in\gamma_{vp}$ 
such that $z\in\overline{vr_0}$ 
(the argument would be the same for $\gamma_{vq},\gamma_{wp},\gamma_{wq}$) 
and so $z=(1-\lambda)v + \lambda r_0$, $\lambda\in(0,1)$. 
Since $r_0\in\gamma_{vp}\subset R(v)$, it is $D(v,r_0) = \displaystyle{\min_{u\in V}
D(u,r_0)}$. 
However, $z=(1-\lambda)v + \lambda r_0$, $\lambda\in(0,1)$, 
	implies that $D(z,r_0) = \lambda D(v,r_0) < D(v,r_0)$, a contradiction.

$B$ is the interior of the triangles $\widetriangle{vwp}$ and
$\widetriangle{vwq}$. We will only show that $z\notin\mathbf{int\ }{\widetriangle{vwp}}$, but the 
exact same argument proves that $z\notin \mathbf{int\ }{\widetriangle{vwq}}$.
Because $p\in R(v)\cap R(w)$, $p$ is closest, and equidistant 
to $v,w$, and therefore the open ellipse $\theta_p(v)$ cannot contain any site (or else $p$ would
belong to its corresponding Voronoi region instead). 
If some site $z$ is $z\in\mathbf{int\ }{\widetriangle{vwp}}$ then, since
$\mathbf{int\ }{\widetriangle{vwp}}\subset\theta_p(v)$, it is 
$z\in\mathbf{int\ }{\widetriangle{vwp}}\subset\theta_p(v)$, a
contradiction. 

Since $z\notin A\cup B$, and $\mathbf{int\ }{I} \subseteq A\cup B$, then it is
$z\notin \mathbf{int\ }{I}$, which is the contradiction that we were after in the first
place, and therefore no two edges in $\tilde{G}$ connect the same vertices. 
\end{proof}

\noindent{\bf Lemma~\ref{lem:degen}}
\emph{
	$\bar{G}$ has no degenerate (null area) elements. }
\begin{proof}
Because $G$ is dual to $P$, to every degenerate face of $\bar{G}$ with 
vertices $u,v,w$ corresponds a 
primal vertex in $\tilde{P}$: a point $c\in
R(u)\cap R(v)\cap R(w)$, and therefore
all of $u,v,w$ are in the boundary of the ellipse 
$\theta_c(u)$ (with $\theta_c(u)=\theta_c(v)=\theta_c(w)$). 

Since a line can intersect an ellipse at most at two points, 
three or more points cannot be both colinear and in $\partial\theta_c(u)$. 
If $u,v,w$ are in a degenerate face, they are colinear, and in $\partial\theta_c(u)$, a contradiction. 
Therefore no faces of $\bar{G}$ are degenerate. 
\end{proof}


\section*{Appendix E}\label{app:boundary}

\noindent{\bf Lemma~\ref{boundary_easy}}
\emph{	
To every boundary edge $(v_i,v_j)$ of $G$ corresponds a segment in the boundary of $\mathcal{CH}(V)$.  \emph{$[B\subseteq\mathcal{B}]$}
}
\begin{proof}
By the definition of $G$, 
to every boundary edge $(v_i,v_j)\in B$ corresponds a primal edge in $P$
that is incident to the point at infinity $p_\infty\in V_p$. 
In turn, to this edge corresponds an edge in $\tilde{P}$: 
an unbounded set $\tilde{E}_{v_i,v_j}$ of points closest to $v_i,v_j$.

Consider the two open half-planes $H^{+}_{ij}$ and $H^{-}_{ij}$ on either
side of the supporting line $l_{ij}$ of $v_i,v_j$. 
Since, by Lem.~\ref{lem:midpoint}, it is $\tilde{E}_{v_i,v_j}\cap l_{ij} \subseteq
\{(v_i+v_j)/2\}$, it must be either that $\tilde{E}_{v_i,v_j}\cap H^{+}_{ij}$ or
$\tilde{E}_{v_i,v_j}\cap H^{-}_{ij}$ are unbounded. 
By Lem.~\ref{lem:halfspace}, they cannot both be, or else 
$H^{+}_{ij}\cap V=\phi$ and $H^{-}_{ij}\cap V=\phi$, and therefore all sites
would be in $l_{ij}$ (all colinear). 
Assume w.l.o.g.\  that $\tilde{E}_{v_i,v_j}\cap H^{+}_{ij}$ is unbounded. 

By Lem.~\ref{lem:halfspace}, $\tilde{E}_{v_i,v_j}\cap H^{+}_{ij}$ unbounded implies 
$H^{+}_{ij}\cap V=\phi$, and therefore
$v_i,v_j\in W$ and $\overline{v_i v_j}\subseteq\partial\mathcal{CH}(V)$  
($v_i,v_j$ are vertices in the boundary of $\mathcal{CH}(V)$). 

It only remains to show that $v_i,v_j$ are consecutive in the sequence $(w_i :
i=1,\dots,m)$. 
We prove this by contradiction. 
If they are not consecutive, since $\overline{v_i
v_j}\subseteq\partial\mathcal{CH}(V)$, 
there must be a site $w\in \overline{v_i v_j}$, $w\neq v_i,v_j$. 
Pick some
$p\in \tilde{E}_{v_i,v_j}$, by definition closest to $v_i,v_j$. 
By the convexity of $D(\cdot,p)$, it is $D(w,p) < D(v_i,p)=D(v_j,p)$, a contradiction. 

Since $v_i,v_j$ are consecutive vertices in $(w_i : i=1,\dots,m)$, then 
$(v_i,v_j)\in\mathcal{B}$. 
\end{proof}

\noindent{\bf Lemma~\ref{lem:mij}}
\emph{
	If $\pi(p)=(w_i+w_j)/2$, with $p\in\mathbb{R}^2$, $w_i,w_j\in W$, then
$D(w_i,p) = D(w_j,p)$. }
\begin{proof}
	Let $m_{ij}=(w_i+w_j)/2$. If $x(\lambda) = (1-\lambda)w_i + \lambda w_j$, 
then the fact that $\pi(p)=m_{ij}$ means that $m_{ij}=x(1/2)$ is closest to $p$
in the support line of $w_i,w_j$. Therefore, $\frac{d D(x(\lambda),p)^2}{d\lambda}(1/2) =
0$, or equivalently:
\begin{eqnarray*}
 2\frac{d D(x(\lambda),p)^2}{d\lambda}(1/2) &=&
		(M_p w_i + M_p w_j - M_p p)^t  (M_p w_j - M_p w_i) \\
	&=& (M_p w_j - M_p p)^t (M_p w_j - M_p p) - (M_p w_i - M_p p)^t (M_p w_i - M_p p) \\
	&=& D(w_j,p)^2 - D(w_i,p)^2 = 0
\end{eqnarray*}
and thus $D(w_i,p) = D(w_j,p)$. 
\end{proof}

\noindent{\bf Lemma~\ref{lem:contrad}}
\emph{
	There is $\rho$ such that, 
	for any segment $(w_i,w_j)\in\mathcal{B}$, 
	every $p\in H^{-}_{ij}$ with $\|p\| > \rho$
    whose closest point in $l_{ij}$ is $m_p\in\overline{w_i w_j}$ is 
	closer to $V\setminus\{w_i,w_j\}$ than to $l_{ij}$. }
\begin{proof}


For each choice of $(w_i,w_j)\in\mathcal{B}$, define $S_{ij}$ to be the set of
points $p\in H^{-}_{ij}$ whose closest point in $l_{ij}$ is
$m_p\in\overline{w_i w_j}$. 
Pick some $v\in H^{-}_{ij}\cap V\subset V\setminus\{w_i,w_j\}$, which always exists
since not all sites are colinear. 
We can use Lem.~\ref{lem:tech}, 
where $w_i,w_j,l_{ij},S_{ij}$ take the role of $v,w,l,S$, 
to conclude that 
there is a sufficiently large $\rho_{ij}$ such that all $p\in S_{ij}$ with
$\|p\|>\rho$ are closer to $v$ than to $m_p$. 
By defining $\rho$ as the maximum of $\rho_{ij}$ over all
$(w_i,w_j)\in\mathcal{B}$, the lemma follows. 
\end{proof}

\noindent{\bf Lemma~\ref{lem:Sn}}
\emph{
	Every continuous function $F:\mathbb{S}^n\rightarrow\mathbb{S}^n$ that is not onto has a fixed point. 
}
\begin{proof}
	Assume $F$ misses $p\in\mathbb{S}^n$, and let
$\gamma:\mathbb{S}^n\setminus\{p\}\rightarrow D^n$ be a diffeomorphism
between the punctured sphere and the open unit disk. 
Since $\gamma\circ F$ is continuous and $\mathbb{S}^n$ is compact,
then the set $C = (\gamma\circ F) (\mathbb{S}^n)\subset D^n$ is compact.

The function $g:C\rightarrow C$ with $g = \gamma \circ F\circ\gamma^{-1}$ 
is continuous and therefore, by Brouwer's fixed point theorem~\cite{Milnor}, has a fixed point $x\in C$. 
The fact that $(\gamma \circ F\circ \gamma^{-1}) (x) = x$ implies $F(\gamma^{-1}(x)) = \gamma^{-1}(x)$ 
and thus $\gamma^{-1}(x)\in\mathbb{S}^n$ is a fixed point of $F$. 
\end{proof}

\noindent{\bf Lemma~\ref{lem:hard}}
\emph{
To every segment $(w_i,w_j)$ in the boundary of
$\mathcal{CH}(V)$ corresponds a boundary edge of $G$.  \emph{$[B\supseteq\mathcal{B}]$}
}
\begin{proof}  
Pick a sufficiently large $\rho > \max_{v\in V}\|v\|$ such that every $p$ with
$\|p\| > \rho$ is outside $\mathcal{CH}(V)$ and such that Lemmas~\ref{lem:VW} and~\ref{lem:contrad} hold. 
For any $\sigma > \rho$, if $A:C(\sigma)\rightarrow C(\sigma)$ is the antipodal map $A(p) = -p$, 
then, by continuity of $\pi$ and 
by the continuity of $\nu_\sigma$, the function $A\circ\nu_\sigma\circ\pi:C(\sigma)\rightarrow C(\sigma)$ is
continuous. 
By Lemmas~\ref{lem:mij} and~\ref{lem:VW}, if for some $p_{ij}\in C(\sigma)$ it is 
$\pi(p_{ij})= (w_i+w_j)/2$ with
$(w_i,w_j)\in\mathcal{B}$, then $p_{ij}$ is (strictly) closest to
$w_i,w_j$, and therefore belongs to the primal edge $\tilde{E}_{w_i,w_j}$, which implies
that $(w_i,w_j)\in B$. 

Showing that $\mathcal{B}\subseteq B$ now reduces to showing that for all 
$(w_i,w_j)\in\mathcal{B}$, for all $\sigma > \rho$, there is $p_{ij}\in
C(\sigma)$ such that
$\pi(p_{ij})= (w_i+w_j)/2$. 

Assume otherwise: that there is $(w_i,w_j)\in\mathcal{B}$ such that no $p\in
C(\sigma)$ satisfies $\pi(p) = (w_i+w_j)/2$. Then the function 
$A\circ\nu_\sigma\circ\pi:C(\sigma)\rightarrow C(\sigma)$ is not onto and therefore, 
by Lem.~\ref{lem:Sn} (and using the fact that $C(\sigma)$ is isomorphic to $\mathcal{S}^1$), 
it must have a fixed point $q$. 

Since $(A\circ\nu_\sigma\circ\pi)(q) = q$ then $(\nu_\sigma\circ\pi)(q) = -q$. 
Since $x=\pi(q)$ is the closest point to $q$ in $\partial\mathcal{CH}(V)$, 
	there is a segment $(w_k,w_l)\in\mathcal{B}$ such that
$x\in\overline{w_k w_l}$. Consider two open half spaces
$H^{+}_{kl}$ and $H^{-}_{kl}$ on either side of the supporting line of
$w_k,w_l$. Since not all sites are colinear, we can choose these half spaces so that 
$H^{+}_{kl}\cap V=\phi$ and $H^{-}_{kl}\cap V \neq \phi$. 
By the definition of $\nu_\sigma$, and recalling that the chosen origin of
$\mathbb{R}^2$ is in the interior $\mathbf{int\ }{\mathcal{CH}(V)}$ of the convex hull, 
it is $\nu_\sigma(x)\in H^{+}_{kl}$, and $q=-\nu_\sigma(x)\in H^{-}_{kl}$. 
To see this note that the outward-facing normal $n(x)$ is defined so that
$x+n(x)\in H^{+}_{kl}$ 
and so $\nu_\sigma(x) = \sigma\cdot n(x) / \|n(x)\| \in  H^{+}_{kl}$.
On the other hand, since the origin is in
$\mathbf{int\ }{\mathcal{CH}(V)}$, the fact that $\nu_\sigma(x)\in H^{+}_{kl}$ 
implies $q=-\nu_\sigma(x)\in  H^{-}_{kl}$. 

Since $\rho$ was chosen sufficiently large for Lem.~\ref{lem:contrad} to
hold, and $q\in  H^{-}_{kl}$, $q$ is closer to some site $v\in
V\setminus\{w_k,w_l\}$ than to $\overline{w_k w_l}$. 
Since $v\in\mathcal{CH}(V)$, this contradicts the fact
that $x=\pi(q)\in\overline{w_k w_l}$ is the closest point to $q$ in
$\mathcal{CH}(V)$, and thus $\mathcal{B}\subseteq B$.
\end{proof}


\section*{Appendix F}\label{app:interior}

\noindent{\bf Lemma~\ref{lem:non-negative}}
\emph{
	Given a non-vanishing one-form $\xi^n$, the sum of indices of interior vertices of $\bar{G}$ is non-negative. 
}
\begin{proof}
Given non-vanishing $\xi^n$, 
the index of a face $f$ is $\mathbf{ind}_{_{\xi^n}}(f) \le 0$. 
To see this, assume otherwise: a face with vertices $v_1,\dots,v_m$ around it, and index one satisfies, by the
definition of index and of $\xi^n$, 
$n^t v_1 < \dots < n^t v_m < n^t v_1$ 
(or $n^t v_1 > \dots > n^t v_m > n^t v_1$), a
contradiction. 

Because, by Lem.~\ref{col:simple-boundary}, the boundary edges of $\bar{G}$ form a (not necessarily
strictly) convex, simple polygonal chain then, given any non-vanishing $\xi^n$, all the boundary vertices have
index zero, except for the ``topmost" ($\underset{v\in V}{\mathbf{argmax\ }} \xi^n(v)$) 
and ``bottommost" ($\underset{v\in V}{\mathbf{argmin\ }} \xi^n(v)$) vertices, which have
index one. 

Since face indices are non-positive, and the sum of indices of
boundary vertices is two then, by Lem.~\ref{lem:ph},
the sum of indices of {interior} vertices must be non-negative. 
\end{proof}

%

\noindent{\bf Lemma~\ref{lem:index-1}}
\emph{
If $\bar{G}$ has an edge foldover then there is a non-vanishing one-form $\xi^n$ such
that $\mathbf{ind}_{_{\xi^n}}(v) < 0$ for some interior vertex $v\in V\setminus W$. 
}
\begin{proof}
If edge foldover $e=(v,w)$ is a non-boundary edge, then at least one
of its incident vertices, say $v$ is an interior vertex $v\in
V\setminus W$. 

Consider the two faces $f_1,f_2$ incident to $e$, which, by definition of
edge foldover, are on the same side of its
supporting line, and the two edges $e_1,e_2$ in $f_1,f_2$ respectively,
incident to $v$. 
Taking the half-line $h$ from $v$ towards $w$ as reference, consider 
the (open) set $L_i\subset\mathbb{S}^1$ of directions ranging from $h$ to
$e_i$. 
The set $L=L_1\cap L_2$ is not empty since, by Lem.~\ref{lem:degen},
$f_1,f_2$ are not degenerate, and therefore neither $e_1,e_2$ are parallel
to $h$. $L$ is also uncountable, since it is a range of the form
\[ L = \{n\in\mathbb{S}^1 : 
			{n_{\perp}}^t {h} < 0 \wedge
			{n_{\perp}}^t {e}_1 > 0 \wedge
			{n_{\perp}}^t {e}_2 > 0\}\]
where ${h},{e}_i$ are the direction vectors of $h,e_i$, and
${n_\perp}$ is one of the two orthogonal directions to $n$, 
chosen to fit the definition. 

Because $L$ is not empty, and it is uncountable, and because the set of edges $E$ is finite,
then there is always some direction $n\in L$ that is not orthogonal to any edge in
$E$. We prove that the non-vanishing one-form $\xi^n$ is such that
$\mathbf{ind}_{_{\xi^n}}(v)<0$. 

The (cyclic) sequence of oriented half-edges {around} $v$ is, without loss of generality,
$\mathcal{S}=\left[(v,v_1);(v,w);(v,v_2);\dots\right]$, and therefore the values of the
one-form around $v$ are $[\xi^n(v_1)-\xi^n(v)$, $\xi^n(w)-\xi^n(v)$,
$\xi^n(v_2)-\xi^n(v)$, $\dots]$. 
By the definition of $n$, it is $\xi^n(v_1)<\xi^n(v)$, $\xi^n(w)>\xi^n(v)$,
and $\xi^n(v_2)<\xi^n(v)$, and therefore 
the number of sign changes in 
the subsequence
$\mathcal{S}'=[(v,v_1);(v,w);(v,v_2)]$ is four. 
Since the number of sign changes in the full sequence $\mathcal{S}$ cannot
be less
than that of its subsequence $\mathcal{S}'$, 
it is $\mathbf{sc}_{_{\xi^n}}(v)>4$ and therefore $\mathbf{ind}_{_{\xi^n}}(v)=1 - \mathbf{sc}_{_{\xi^n}}(v)/2 <0$. 
\end{proof}

\noindent{\bf Lemma~\ref{lem:index1}}
\emph{
Given $n\in\mathbb{S}^1$ and non-vanishing one-form $\xi^n$, if $\bar{G}$
has an interior vertex $v\in V\setminus W$ with index
$\mathbf{ind}_{_{\xi^n}}(v)=1$, then there is a face $f$ of
$G$ that does not satisfy the empty circum-ellipse 
property.  
%
}
\begin{proof}
We must prove that there is a face $f$ all of whose circumscribing ellipses
contain some vertex in its interior. 

Consider the vertex $v\in V\setminus W$ with
$\mathbf{ind}_{_{\xi^n}}(v)=1-\mathbf{sc}_{_{\xi^n}}(v)=1$, and thus
with $\mathbf{sc}_{_{\xi^n}}(v)=0$. 
If $\left[u_1,u_2,\dots,u_m\right]$ is the cyclic sequence of vertices neighboring
$v$, then 
$\mathbf{sc}_{_{\xi^n}}(v)=0$ implies either $\xi^n(u_i)>\xi^n(v)$, $i=1,\dots,m$, or $\xi^n(u_i)<\xi^n(v)$, $i=1,\dots,m$. 
Assume the former w.l.o.g. 
The line $l=\{x\in\mathbb{R}^2 : n^t x = n^t v\}$, passing through
$v$, strictly separates $v$ from the convex hull of its neighbors. 

Consider the mesh $\bar{G'}$, with the same structure as $\bar{G}$ but in which all the incident faces to $v$
are eliminated. 
We show that, in $\bar{G'}$, the face count of $v$ (the number of faces in which $v$ lies) is at
least one. 
Since $l$ separates $v$ from its neighbors, it also separates all the faces
incident to $v$ from $v$ (except for $v$ itself, which lies on $l$). 
Pick any direction $d\in\mathcal{S}^1$ with $n^t{d} < 0$. 
The half-line $h$ starting at $v$ with direction $d$ 
does not intersect any face in $G$ that is incident to $v$. 
Since there is only a finite number of edges and vertices, it is always
possible to choose $h$ not to contain any vertex other than $v$, 
and not to be parallel to any edge. 
Since $\mathcal{CH}(V)$ is bounded and $h$ isn't, there is
some point $x\in h$ outside $\mathcal{CH}(V)$, whose face count must be zero. 
Moving from $x$ toward $v$, $h$ crosses $\partial\mathcal{CH}(V)$ only once 
(since $\mathcal{CH}(V)$ is convex), incrementing the face count to one. 
Because every interior edge is incident to exactly two faces, 
every subsequent edge cross (which is transversal because $h$ is not
parallel to any edge) modifies the face count by either zero, two, or
minus two. Since the face count cannot be negative, and it is one at 
$h\cap \partial\mathcal{CH}(V)$, then it must be at least one at $v$. 
Since $\bar{G'}$ does not contain any face incident to $v$, 
this implies that there is
some face $f$ not incident to $v$ such that $v\in f$. 

We prove that the face $f$ above cannot satisfy the ECE property. 
Since $v$ is in $f$ but is not incident to it, and $f$ is convex
then, by Carath\'eodory's theorem, $v$ can be written as a 
convex combination $v=\lambda_1 u_1+\lambda_2 u_2+\lambda_3 u_3$,
$\sum_{i=1}^3\lambda_i=1$, $\lambda_i\in(0,1)$ of vertices $u_1,u_2,u_3$
incident to $f$. 
Given an ellipse $\theta$ circumscribing the vertices incident to $f$, 
because $\theta$ is convex, and $u_1,u_2,u_3$ lie in its boundary, 
then any convex combination of them with $\lambda_i\in(0,1)$ must be in the
interior of $\theta$, and therefore $f$ does not satisfy the ECE property. 
\end{proof}

\noindent{\bf Lemma~\ref{lem:ef}}
\emph{
$\bar{G}$ 
has no edge foldovers. }
\begin{proof}
Assume $\bar{G}$ has an edge foldover. 
By Lem.~\ref{lem:index-1}, there is a non-vanishing one-form
$\xi^n$ such that some interior vertex $v\in V\setminus W$ has $\mathbf{ind}_{_{\xi^n}}(v)<0$. 
Since, by Lem.~\ref{lem:non-negative}, the sum of indices of interior vertices is
non-negative,
then there must be
at least one interior vertex $u\in V\setminus W$ with positive index
$\mathbf{ind}_{_{\xi^n}}(u)=1$. 
In that case, by Lem.~\ref{lem:index1}, there is a face of $\bar{G}$ that does not satisfy the
ECE property, raising a contradiction. 
Therefore $\bar{G}$ has no edge foldovers. 
\end{proof}

\noindent{\bf Lemma~\ref{lem:main-weak}}
\emph{
If its (topological) boundary is simple and closed, then the straight-line dual $\bar{G}$ of an orphan-free diagram, with vertices incident to at most three sites, is an embedded triangulation. 
}
\begin{proof}

Given a point $p\in\mathbf{int\ }{\mathcal{CH}(V)}$ in
the interior of the convex hull of $V$, we show that its \emph{face count}
(the number of faces with straight edges that contain it) is one. 
Consider a line $l$ passing through $x$ that does not pass through any vertex of
$\bar{G}$, and is not parallel to any (straight) edge. 
It is always possible to find such a line since the set of vertices and edges is finite. 
Because the line is unbounded and $\mathcal{CH}(V)$ is bounded, there is a
point $x\in l$ that is outside $\mathcal{CH}(V)$. At this point clearly the
face count is zero. 
Moving from $x$ toward $p$, $l$ crosses the boundary of $\mathcal{CH}(V)$
(and therefore, by Thm.~\ref{thm:boundary}, the boundary of $\bar{G}$) 
only once, since it is a simple convex polygonal chain, incrementing the face
count by one. At every edge crossing (which is transversal by the choice of
line), the face count remains one since, by Lem.~\ref{lem:ef} there are no
edge foldovers, and thus every non-boundary edge is incident to two faces
that lie on either side of its supporting plane. Therefore the face count at
$p$ must be one. 

Since every point inside $\mathcal{CH}(V)$ is covered once by faces in
$\bar{G}$,
and the boundaries of $\bar{G}$ and $\mathcal{CH}(V)$ coincide, then
$\bar{G}$ is a single-cover
of $\mathcal{CH}(V)$. 
Because two straight edges that cross at a non-vertex always generate points with
face count higher than one, then the edges of $\bar{G}$ can only meet at vertices, and
therefore $\bar{G}$ is embedded. 
\end{proof}


\section*{Appendix G}\label{app:generic}

	So far we have assumed that Voronoi vertices are incident (equidistant) to no more than three sites. 
	In general, however, Voronoi vertices may be incident to three or more sites, and the straight-edge dual 
	$\bar{G}$ will be a polygonization, instead of a triangulation (simplicial complex). 
	We show here that, even in this, more general case, the polygonization $\bar{G}$ is embedded, and can 
	be easily triangulated into  an embedded simplicial complex. 
	The argument is quite simple and is given in summary. \\

\noindent{\bf Theorem~\ref{th:main}}
\emph{
If its (topological) boundary is simple and closed, then the straight-line dual $\bar{G}$ of an orphan-free diagram is an  embedded polygonal mesh with convex faces. 
}
\begin{proof}
	First assume that only Voronoi vertex $c$ is incident to more than three sites $\{v_i : i=1\dots m\}$, $m>3$, 
	ordered in, say, clockwise order around $c$ (which is possible since they are equidistant to $c$ and therefore 
	lie at the boundary of the ellipse $\theta_c(v_1)$). 
	Note that, since the $\{v_i : i=1\dots m\}$ are equidistant to $c$, then any polygon (or triangle) connecting any of the $\{v_i\}$ 
		satisfies the empty circum-ellipse property (with empty circum-ellipse $\theta_c(v_1)$).  
	The  Voronoi regions $\{R_i : i=1,\dots,m\}$, corresponding to those sites, are the only ones incident to $c$. 
	We can order these regions in clockwise order around $c$: $\{R_{i_k} : k\in K\}$ 
	(for some index sequence $K$ with $k\in K$ satisfying $k\in\{1,\dots,m\}$), 
	and therefore the polygon $\Pi$ dual to $c$ will connect vertices $\{v_{i_k} : k\in K\}$, in that order. 
	Now choose any $v_j$ in $\Pi$ and triangulate $\Pi$ in a fan arrangement centered at $v_j$. 
	As pointed out above, both $\Pi$ and every triangle in the triangulation of $\Pi$ satisfies the empty circum-ellipse property. 
	Therefore the new, triangulated $\bar{G}$ (which replaces $\Pi$ by its triangulation), 
	satisfies all the conditions needed for Theorems~\ref{th:main} and~\ref{thm:boundary} (as well as Lem.~\ref{lem:degen}) to hold, 
	and in particular, it is embedded. 

	The claim now easily follows: any closed polygon with vertices $\{v_{i_k} : k\in K\}$, 
		such that all of its fan-arranged triangulations centered at every one of its vertices is embedded, must be convex 
		(given a polygon with non-convex vertex $w$, triangulate it in a fan centered at a vertex incident to $w$: this triangulation is not embedded). 
	Lastly, since vertices $\{v_i : i=1\dots m\}$ are equidistant from $c$, the only polygon connecting them 
		that is also convex is the one where vertices are arranged in clockwise (counter-clockwise) order around $c$. 

	The argument has been made for a single Voronoi vertex incident to more than three sites, but the same argument applies to diagrams in which there are more such vertices 
	(the key being that the triangulations of each dual polygon are independent of one another). 
\end{proof}

\section*{Appendix H}\label{app:noanisotropy}

We show here that we may drop the bounded anisotropy condition from Thm.~\ref{gamma}, at the expense of loosing the property of covering the convex hull of the sites; 
thus concluding that the dual of an orphan-free anisotropic diagram is an embedded polygonal mesh with convex faces. The proof proceeds by reducing this case to that of Thm.~\ref{gamma}. Because of its simplicity, it is given here in summary. 

Assume given a continuous metric $Q$ defined over $\mathbb{R}^2$, and a set $V$ of sites forming an orphan-free anisotropic Voronoi diagram $V_Q$.
For now, assume that the set of Voronoi vertices is finite, and therefore bounded. 
Since there is $\rho > 0$ such that all Voronoi vertices are inside the origin-centered ball $B(0;\rho)$ of radius $\rho$, we can construct a new metric 
$Q'(p) = Q(p)\cdot(1-\lambda(p)) + I\cdot\lambda(p)$, where 
\[ \lambda(p) = \begin{cases}0,&\text{ if } \|p\| < \rho\\ \|p\|-\rho ,&\text{ if }\rho\le\|p\|\le 2\rho\\ 1 ,&\text{ if }\|p\| > 2\rho\end{cases} \]

Clearly, $Q'$ is continuous, since both $Q$ and $\lambda$ are. Since it is $Q'=I$ outside of $B(0;2\rho)$, and $B(0;2\rho)$ is compact, then 
$Q'$ has bounded anisotropy (ratio of eigenvalues). 
We may then apply Thm.~\ref{gamma} to conclude that the dual $\bar{G}_{Q'}$ of $V_{Q'}$ is an embedded polygonal mesh with convex faces. 
Since $B(0;\rho)$ includes all Voronoi vertices of $V_Q$, and it is $Q'=Q$ inside $B(0;\rho)$, then all Voronoi vertices in $V_Q$ are also in $V_{Q'}$. 
Therefore, by duality, all faces of $\bar{G}_{Q}$ are also in $\bar{G}_{Q'}$. Finally, since $\bar{G}_{Q'}$ is embedded with convex faces, 
and removing faces from a polygonal mesh preserves both properties, then $\bar{G}_{Q}$ is also embedded with convex faces. 

The only assumption we have made is that the set of Voronoi vertices of $V_Q$ is finite,  
which can be justified as follows. 

First, we show that Voronoi vertices cannot be isolated, and are always incident to some Voronoi edge. Consider a Voronoi vertex $c_{ijk}$ closest and equidistant to sites $v_i,v_j,v_k$ (the reasoning for vertices equidistant to more sites is analogous). 
Since $Q$, and therefore $D$, is continuous, there is a (possibly small) closed ball $B(c_{ijk},\epsilon)$ centered at $c_{ijk}$ where all points are \emph{strictly} closer to $v_i,v_j,v_k$ than to any other site, and therefore in $B(c_{ijk},\epsilon)$, the Voronoi diagram of $V$ is the same as that of $\{v_i,v_j,v_k\}$. 
We consider the Voronoi diagrams $\{v_i,v_j\}_Q$ and $\{v_j,v_k\}_Q$ of $\{v_i,v_j\}$ and $\{v_j,v_k\}$ (restricted to  $B(c_{ijk},\epsilon)$), respectively. 
By Corollary~\ref{cor:Pij}, the only Voronoi edge of $\{v_i,v_j\}_Q$ is connected, and likewise for $\{v_j,v_k\}_Q$. Since $c_{ijk}$ is in their intersection, then it must be incident to both, and therefore it is not isolated.

By Corollary~\ref{cor:Pij}, the edges (connected sets of points equidistant to two given sites) of $V_Q$ are unique (thus finite since $V$ is finite), and by the induction argument of Sec.~\ref{sec:setup}, each edge connects two Voronoi vertices. If there are $n_e$ Voronoi edges then, since Voronoi vertices are always incident to some Voronoi edge, and each edge connects two vertices, there cannot be more than $2 n_e$ Voronoi vertices. In particular, the number of Voronoi vertices is finite. This concludes the proof. 

Finally, note that, from the way we have defined the dual $\bar{G}_Q$, and the fact that it is embedded, we can conclude that Voronoi vertices are unique (otherwise, multiple Voronoi vertices would result in coincident dual polygons, which would contradict the fact that $\bar{G}_Q$ is embedded). 
This, together with Corollary~\ref{cor:Pij}, and the  orphan-freedom assumption means that, as stated in Corollary~\ref{uniqueVD}, orphan-freedom is sufficient to ensure that all the elements (vertices, edges, faces) of the primal diagram are unique.

%
%
%

\end{document}